\documentclass[a4paper,USenglish]{lipics-v2021}

\def\REVIEW{0}
\def\SAVESPACE{0}
\def\ANONYMOUS{0}
\nolinenumbers  
\hideLIPIcs     

\usepackage[ruled,linesnumbered,vlined,boxed,commentsnumbered]{algorithm2e}
\usepackage{comment}
\usepackage{caption}
\usepackage{color}
\usepackage{graphicx}
\usepackage{microtype}
\usepackage{ragged2e}
\usepackage{float}
\usepackage{svg}
\usepackage[normalem]{ulem}
\usepackage{todonotes}

\DeclareMathOperator*{\argmax}{arg\,max}

\def\lf{\tiny}
\newcounter{mylinenumber}

\def\nnll{\refstepcounter{mylinenumber}\lf\themylinenumber}

\newcommand{\commentline}[1]{\hspace{0.05cm}\{ \textit{#1} \}}

\usepackage{url}
\usepackage{cite}

\newcommand{\Name}{\textsc{Pastro}\xspace} 

\def\code#1{\textnormal{\texttt{#1}}}

\theoremstyle{definition} 

\newcommand{\citep}{\cite}

\if \SAVESPACE 1
    \usepackage{titlesec}
    \titlespacing\section{0pt}{7pt}{6pt}
\fi

\usepackage[\if \REVIEW 1 enable \else apply,nonotes \fi]{xnotes}
\AddXNotesUser{pk}{PK}{green}
\AddXNotesUser{ya}{YA}{magenta}
\AddXNotesUser{pp}{PP}{orange}
\AddXNotesUser{at}{AT}{blue}

\newcommand{\true}{\textit{true}}
\newcommand{\false}{\textit{false}}

\newcommand{\ignore}[1]{}

\newcommand{\cC}{\mathcal C}
\newcommand{\dD}{\mathcal D}

\newcommand{\tT}{\mathcal T}

\newcommand{\Lat}{\mathcal L}
\newcommand{\Nat}{\mathbb N}

\newcommand{\myparagraph}[1]{\subparagraph*{#1}}

\title{Permissionless and Asynchronous Asset Transfer}
\subtitle{[Technical Report]}

\if \ANONYMOUS 1
    \author{Anonymous author(s)}{Anonymous affiliation}{}{}{}
    \authorrunning{Anonymous authours}
    \Copyright{Anonymous author(s)}
\else
    \author{Petr Kuznetsov}{LTCI, T\'el\'ecom Paris, Institut Polytechnique de Paris}{petr.kuznetsov@telecom-paris.fr}{}{}
    \author{Yvonne-Anne Pignolet}{DFINITY}{yvonneanne@dfinity.org}{}{}
    \author{Pavel Ponomarev}{ITMO University}{pavponn@gmail.com}{}{}
    \author{Andrei Tonkikh}{National Research University Higher School of Economics}{andrei.tonkikh@gmail.com}{}{}
    
    \authorrunning{P. Kuznetsov, Y.-A. Pignolet, P. Ponomarev, A. Tonkikh}
    
    \Copyright{P. Kuznetsov, Y.-A. Pignolet, P. Ponomarev, A. Tonkikh}
\fi

\begin{document}



\maketitle

\begin{abstract}
Most modern asset transfer systems use \emph{consensus} to maintain a totally ordered chain of transactions. 
It was recently shown that consensus is not always necessary for implementing asset transfer. 
More efficient, \emph{asynchronous} solutions can be built using \emph{reliable broadcast} instead of consensus. 
This approach has been originally used in the closed (permissioned) setting. 
In this paper, we extend it to the open (\emph{permissionless}) environment.
We present {\Name}, a permissionless and asynchronous asset-transfer implementation, in which \emph{quorum systems}, traditionally used in reliable broadcast, are replaced with a weighted \emph{Proof-of-Stake} mechanism. 
{\Name} tolerates a \emph{dynamic} adversary that is able to adaptively corrupt participants based on the assets owned by them \atremove{into account}.
\end{abstract}

\section{Introduction}
\label{sec:intro}
Inspired by advances in peer-to-peer data replication~\citep{bitcoin,ethereum}, a lot
of efforts are currently invested in designing an algorithm for consistent and 
efficient exchange of assets in \emph{dynamic} settings, where the set of 
participants, actively involved in processing transactions, varies over time.

Sometimes such systems are called \emph{permissionless}, emphasizing the fact that they assume
no trusted mechanism to regulate who and when can join the system. 
In particular, permissionless protocols should tolerate the notorious 
\emph{Sybil attack}~\citep{sybil}, where the adversary creates an unbounded number of ``fake'' identities.

Sharing data in a permissionless system is a hard problem. 
To solve it, we have to choose between consistency and efficiency. 
Assuming that the network is synchronous and that the adversary can only possess 
less than half of the total computing power, Bitcoin~\citep{bitcoin} and 
Ethereum~\citep{ethereum} make sure that participants reach  \emph{consensus} on the order in which they access and modify the  shared data.       
For this purpose, these systems employ a \emph{proof-of-work} (PoW) mechanism to artificially slow down active participants (\emph{miners}).
%
The resulting algorithms are notoriously slow and waste tremendous amounts of energy.

Other protocols obviate the energy demands using \emph{proof-of-stake}~\citep{PoS16,ouroboros17,algorand},
\emph{proof-of-space}~\citep{PoSE15}, or \emph{proof of
space-time}~\citep{PoST16}.  
However, these proposals still resort to synchronous networks, randomization, non-trivial cryptography and/or assume a trusted setup system. 

In this paper, we focus on a simpler problem, \emph{asset transfer}, enabling a set of participants to exchange assets across their \emph{accounts}.
It has been shown~\citep{Gup16,at2-cons} that this problem does not require consensus.
Assuming that every account is operated by a dedicated user, there is no need for the users to agree on a total order in which transactions must be processed:
%
one can build an asset transfer system on top of the \emph{reliable broadcast} abstraction instead of consensus. 
Unlike consensus~\citep{FLP85}, reliable broadcast allows for simple 
\emph{asynchronous} solutions, enabling efficient asset transfer 
implementations that outperform their consensus-based counterparts~\citep{astro-dsn}.

Conventionally, a reliable broadcast algorithm assumes a \emph{quorum system}~\citep{quorums,byz-quorums}.
Every delivered message should be \emph{certified} by a \emph{quorum} of participants.
Any two such quorums must have a correct participant in common and in every run, at least one quorum should consist of correct participants only. 
In a static $f$-resilient system of $n$ participants, these assumptions result \atremove{result} in the condition $f<n/3$. 
%
In a \emph{permissionless} system, where the Sybil attack is enabled, assuming a traditional quorum system does not appear plausible. %
Indeed, the adversary may be able to control arbitrarily many identities, undermining any quorum assumptions.  

In this paper, we describe a permissionless asset-transfer system based on \emph{weighted} quorums.
%
More precisely, we replace traditional quorums with  certificates signed by participants holding a sufficient amount of assets, or \emph{stake}.
One can be alerted by this assumption, however: the notion of a ``participant holding stake'' at any given moment of time is not well defined in a decentralized consensus-free system where assets are dynamically exchanged and participants may not agree on the order in which transactions are executed. 
We resolve this issue using the notion of a \emph{configuration}.
A configuration is a partially ordered set of transactions that unambiguously determines the active system participants and the distribution of stake among them. 
As has been recently observed, configurations form a \emph{lattice order}~\citep{rla} and a \emph{lattice agreement} protocol~\citep{gla,rla,bla} can be employed to make sure that participants properly reconcile their diverging opinions on which configurations they are in.

Building on these abstractions, we present the {\Name} protocol to settle asset transfers, despite a \emph{dynamic adversary} that can choose which participants to compromise \emph{during} the execution, taking their current stake into account.
The adversary is restricted, however, to corrupt participants that together own less than one third of stake in any \emph{active candidate} configuration.
Intuitively, a configuration (a set of transactions) is an active candidate configuration if all its transactions \emph{can be} accepted by a correct process.
At any moment of time, we may have multiple candidate configurations, and the one-third stake assumption must hold for each of them.

Note that a \emph{superseded} configuration that has been successfully replaced with a new one can be compromised by the adversary.
To make sure that superseded configurations cannot deceive slow participants that are left behind the reconfiguration process, we employ a \emph{forward-secure} digital signature scheme~\citep{bellare1999forward,drijvers2019pixel}, recently proposed for Byzantine fault-tolerant reconfigurable systems~\citep{bla}. 
The mechanism allows every process to maintain a single public key and a ``one-directional'' sequence of matching private keys: it is computationally easy to compute a new private key from an old  one, but not vice versa.
Intuitively, before installing a new configuration, one should ask holders of $>2/3$ of stake of the old one to upgrade their private keys and destroy the old ones.  

We believe that {\Name} is the right alternative to heavy-weight consensus-based replicated state machines, adjusted to applications that do not require global agreement on the order on their operations~\citep{crdt,rla}.   
Our solution does not employ PoW~\citep{bitcoin,ethereum} and does not rely on complex cryptographic constructions, such as a common coin~\citep{common-coin, algorand}.
More importantly, unlike recently proposed solutions~\citep{abc-tr,byz-bcast}, {\Name} is resistant against a dynamic adversary that can choose which participants to corrupt in a dynamic manner, depending on the execution.  

In this paper, we \yareplace{chose to}{} first present {\Name} in its simplest version, as our primary message is a possibility result: a permissionless asset-transfer system can be implemented in an asynchronous way despite a dynamic adversary.      
We then discuss multiple ways of improving and generalizing our system.
In particular, we address the issues of maintaining a dynamic amount of assets via an inflation mechanism and \atreplace{improving complexity}{the performance} of the system by delegation and incremental updates.  

\myparagraph{Road map.}
We overview related work in Section~\ref{sec:related_work}.
In Section~\ref{sec:model}, we describe our model and recall basic definitions, followed by the formal statement of the asset transfer problem in Section~\ref{sec:crypto}.
In Section~\ref{sec:algo}, we describe {\Name} and outline its correctness arguments in Section~\ref{sec:proof}. 
We conclude \yareplace{by discussing complexity, security and incentive aspects,}{with practical challenges to be addressed} as well as related open questions in Section~\ref{sec:discussion}. 
Detailed proofs \yaadd{and the discussion of \atreplace{enhancements}{optimizations as well as delegation, fees, inflation, and practical aspects of using forward-secure digital signatures}} are deferred to the appendix.

\section{Related Work}
\label{sec:related_work}
The first and still the most widely used permissionless cryptocurrencies are blockchain-based~\citep{coinmarketcap}. 
To make sure that honest participants agree (with high probability) on the order in which blocks of transactions are applied, the most prominent blockchains rely on proof-of-work~\citep{bitcoin,ethereum} and assume synchronous communication.  
The approach exhibits bounded performance, wastes an enormous amount of energy, and its practical deployment turns out to be hardly decentralized (\url{https://bitcoinera.app/arewedecentralizedyet/}).

To mitigate these problems, more recent proposals suggest to rely on the \emph{stake} rather than on energy consumption. 
E.g., in next version of Ethereum~\citep{ethereum_pos}, random \emph{committees} of a bounded number of \emph{validators} are periodically elected, under the condition that they put enough funds at stake. 
Using nontrivial cryptographic protocols (verifiable random functions and common coins), Algorand~\citep{gilad2017algorand} ensures that the probability for a user to be elected is proportional to its stake, and the committee is reelected after each action, to prevent adaptive corruption \atadd{of} the committee. 
%

In this paper, we deliberately avoid reaching consensus in implementing asset transfer, and build atop asynchronous \emph{reliable broadcast}, following~\citep{Gup16,at2-cons}.
It has been recently shown that this approach results in 
a simpler, more efficient and more robust implementation than consensus-based solutions~\citep{astro-dsn}.
However, in that design a static set of of processes is assumed, i.e., the protocol adheres to the \emph{permissioned} model.

In~\citep{byz-bcast}, a reliable broadcast protocol is presented, that allows processes to join or leave the system without requiring consensus.
ABC~\citep{abc-tr} proposes a direct implementation of a cryptocurrency, under the assumption that the adversary never corrupts processes holding $1/3$ or more stake.  
However, both protocols~\citep{byz-bcast,abc-tr} assume a static adversary: the set of corrupted parties is chosen at the beginning of the protocol.   

In contrast, our solution tolerates an adversary that \emph{dynamically} selects a set of corrupted parties, under the assumption that not too much stake is compromised in an \emph{active candidate} configuration.
Our solution is inspired by recent work on \emph{reconfigurable} systems 
that base upon the reconfigurable lattice agreement abstraction~\citep{rla,bla}.
However, in contrast to these \emph{general-purpose} reconfigurable constructions, our implementation is much lighter.
In our context, a configuration is just a distribution of stake, and every new transaction is a configuration update.     
As a result, we can build a simpler protocol that, unlike conventional asynchronous reconfigurable systems~\citep{SKM17-reconf, rla, bla}, does not bound the number of configuration updates for the sake of liveness.  
%

\section{Preliminaries}
\label{sec:model}
\myparagraph{Processes and channels.}
We assume a set $\Pi$ of potentially participating \emph{processes}.
In our model, every process acts both as a \emph{replica} (maintains a local copy of the shared data) and as a \emph{client} (invokes operations on the data).\footnote{We discuss how to split the processes into replicas and clients in \ppreplace{Section~\ref{sec:discussion}}{Appendix~\ref{app:optimizations}}.}
\atreplace{We assume the existence of a global clock with range $\Nat$, but the processes do not have access to it.}{In the proofs and definitions, we make the standard assumption of existence of a global clock not accessible to the processes.}

At any moment of time, a process can be \emph{correct} or \emph{Byzantine}. 
%
%
We call a process \emph {correct} as long as it faithfully follows the algorithm it has been assigned.
A process is \emph{forever-correct} if it \atremove{is} remains correct forever.
A correct process may turn \emph{Byzantine}, which is modelled as an explicit event in our model (not visible to the other processes).
A Byzantine process may perform steps not prescribed by its algorithm or prematurely stop taking steps.
%
Once turned Byzantine, the process stays Byzantine forever.

We assume a \emph{dynamic} adversary that can choose the processes to corrupt (to render Byzantine) depending on the current execution (modulo some restrictions that we discuss in the next section).    
In contrast, a \emph{static} adversary picks up the set of Byzantine processes \emph{a priori}, at the beginning of the execution.   

In this paper we assume that the computational power of the adversary is bounded and, consequently, the cryptographic primitives used cannot be broken.

We also assume that each pair of processes is connected via \emph{reliable authenticated channel}. 
If a forever-correct process $p$ sends a message $m$ to a forever-correct process
$q$, then q eventually receives $m$. 
Moreover, if a correct process $q$ receives a message $m$ \atreplace{from a process $p$ at time $t$, and $p$ is correct at time $t$}{from a correct process $p$}, then $p$ has indeed sent $m$ to $q$\atremove{ before $t$}.

For the sake of simplicity, we assume that the set $\Pi$ of potentially participating processes is finite.\footnote{In practice, this assumptions boils down to requiring that the rate at which processes are added is not too high.
Otherwise, if new stake-holders are introduced into the system at a speed prohibiting a client from reaching a sufficiently large fraction of them, we cannot make sure that the clients' transactions are eventually accepted.}

\myparagraph{Weak reliable broadcast (WRB).} In this paper we assume a \emph{weak reliable broadcast primitive} to be available. 
The implementation of such primitive ensures the following properties:
\begin{itemize}
    \item If a correct process delivers a message $m$ from a correct process $p$, then $m$ was previously broadcast by $p$;
    \item If a forever-correct process $p$ broadcasts a message $m$, then $p$ eventually delivers $m$;
    \item If a forever-correct process delivers $m$, then every forever-correct process eventually delivers $m$.
\end{itemize}

This weak reliable broadcast primitive can be implemented via a gossip protocol~\citep{Gossiping}.

\myparagraph{Forward-secure digital signatures.}
Originally, \emph{forward-secure digital signature schemes}~\citep{bellare1999forward,drijvers2019pixel} were designed to resolve the key exposure problem:
if the signature (private) key used in the scheme is compromised, then the adversary is capable to forge any previous (or future) signature.
Using forward secure signatures, it is possible for the private key to be updated arbitrarily many times with the public key remaining fixed. 
Also, each signature is associated with a timestamp. This helps to identify messages which have been signed with the private keys that are already known to be compromised.

To generate a signature with timestamp $t$, the signer uses
secret key $sk_t$. The signer can update its secret key and get $sk_{t_2}$ from $sk_{t_1}$ if $t_1 < t_2 \leq T$.
However ``downgrading'' the key to a lower timestamp, from $sk_{t_2}$ to $sk_{t_1}$, is computationally infeasible. 
As in recent work on Byzantine fault-tolerant reconfigurable systems~\citep{bla}, 
we model the interface of forward-secure signatures with an oracle which associates every process $p$ with a timestamp $st_p$. The oracle provides the following functions:
\begin{itemize}
    \item $\code{UpdateFSKey}(t)$ sets $st_p$ to $t$ if $t \ge st_p$;
    \item $\code{FSSign}(m, t)$ returns a signature for a message $m$ and a timestamp $t$ if $t \ge st_p$, otherwise~$\perp$;
    \item $\code{FSVerify}(m, p, s, t)$ returns \textit{true} if \atrev{$s \neq \bot$ and it} was generated by process $p$ using $\code{FSSign}(m, t)$, \textit{false}~otherwise.
\end{itemize}

In most of the known implementations of forward-secure digital signature schemes the parameter~$T$ should be fixed in advance, \atreplace{what}{which} makes number of possible private key updates finite. 
At the same time, some forward-secure digital schemes~\citep{MaMiMi02} allow for an unbounded number of key updates ($T = +\infty)$, however the time required for an update operation depends on the number of updates. 
Thus, local computational time may grow indefinitely albeit slowly.
We discuss these two alternative implementations and reason about the most appropriate one for our problem \atreplace{in Section~\ref{sec:discussion}}{in Appendix~\ref{app:optimizations}}.


\myparagraph{Verifiable objects.} 
We often use \textit{certificates} and \textit{verifiable objects} in our protocol description.
We say that object $\textit{obj} \in O$ is \textit{verifiable} in terms of a given \textit{verifying~function} $\code{Verify}: O\times \Sigma_O \rightarrow \{\true, \false\}$, if it comes together with a \textit{certificate} $\sigma_{obj} \in \Sigma_O$, such that $\code{Verify}(\textit{obj}, \sigma_{obj}) = \true$, where $\Sigma_O$ is a set of all possible certificates for objects of set $O$. 
A certificate $\sigma_{obj}$ is \textit{valid} for $\textit{obj}$ (in terms of a given verifying function) iff $\code{Verify}(obj, \sigma_{obj}) = \true$.
The actual meaning of an object ``verifiability'', as well as of a certificate validness, is determined by the verifying function.

\section{Asset Transfer: Problem Statement}
\label{sec:crypto}
Before defining the asset transfer problem formally, we introduce the concepts used later.

\myparagraph{Transactions.} 
A \emph{transaction} is a tuple $tx = (p, \tau, D)$. Here $p \in
\Pi$ is an identifier of a process inside the system that initiates a 
transaction. 
We refer to $p$ as the \emph{owner} of $tx$. 
 The map $\tau$: $\Pi \rightarrow \mathbb{Z}^+_0$
is the \textit{transfer function}, specifying the amount of funds received by every process $q \in \Pi$ from this transaction.\footnote{We encode this map as a set of tuples $(q,d)$, where $q\in\Pi$ and $d>0$ is the amount received sent to $q$ in $tx$.}
$D$ is a finite set of transactions that $tx$ \emph{depends on}, i.e., $tx$ spends the funds $p$ received through the transactions in $D$. 
We refer to $D$ as the \textit{dependency set} 
of $tx$. 

Let $\tT$ denote the set of all transactions.
The function $\textit{value}: \tT \rightarrow V$ is defined as follows: $\textit{value}(tx) = \sum_{q \in \Pi} tx.\tau(q)$.
$\tT$ contains one special ``initial'' transaction $tx_{init} = (\perp, \tau_{init}, \emptyset)$ with  $\textit{value}(tx_{init})=M$ and no sender. 
This transaction determines the initial  distribution of the total amount of ``stake'' in the system (denoted $M$).

For all other transactions it holds that a transaction $tx$ is \emph{valid} if the amount of funds spent equals the amount of funds received: 
$\textit{value}(tx) = \sum_{t \in t.D} t.\tau(tx.p)$.
For simplicity, we assume that $\tT$ contains only valid transactions: invalid transactions are ignored. 
%
%

Transactions defined this way, naturally form a \emph{directed graph} where each transaction is a node and edges represent dependencies. 
%

We say that transactions $tx_i$ and $tx_j$ issued by the same process \textit{conflict} iff the intersection of their dependency sets is non-empty: $tx_i \not = tx_j, tx_i.p = tx_j.p$ and $tx_i.D \cap tx_j.D \not = \emptyset$. 
A correct member of an asset-transfer system does not attempt to \emph{double spend}, i.e., it never issues \yareplace{conflicting transactions which share some dependencies}{transactions which share dependencies and are therefore conflicting}.


Transactions are equipped with a boolean function $\code{VerifySender}: \tT\times\Sigma_{\tT} \to \{\true,\false\}$. A transaction $tx$ is \emph{verifiable} iff it is verifiable in terms of the function
$\code{VerifySender}(tx, \sigma_{tx})$.  Here $\sigma_{tx}$ is a certificate that confirms that $tx$ was indeed issued by its owner process $p$.
One may treat a transaction's certificate as $p$'s digital signature of $tx$. 

\myparagraph{Asset-transfer system.}
An \emph{asset-transfer} system (AS) maintains a partially ordered set of transactions  $T_p$ and exports one operation:
%
$\code{Transfer}(tx, \sigma_{tx})$ adding transaction $tx$ to the set  $T_p$.
Recall that $tx = (p, \tau, D)$, where $p$ is the owner of transaction, $\tau$ is a transfer map, $D$ is the set of dependencies, and $\sigma_{tx}$ is a matching certificate.



In a distributed AS implementation, every process $p$ holds the set of ``confirmed'' transactions $T_p$ it is aware of, i.e., a local copy of the state. \yarev{A transaction is said to be \emph{confirmed} if some correct process $p$ adds $tx$ to its local copy of the state $T_p$. Set $T_p$ can be viewed as the log of all confirmed transactions a process $p$ is aware of.}
Let $T_p(t)$ denote the value of $T_p$ at~time~$t$. 

%
An AS implementation then satisfies the following properties:

\begin{description}
    \item[Consistency:] \atreplace{For every correct process $p$ at time $t$,}{For every process $p$ correct at time $t$,} $T_p(t)$ contains only verifiable and non-conflicting transactions.
    Moreover, \atreplace{for every two correct processes $p$ and $q$ at times $t$ and $t'$}{for every two processes $p$ and $q$ correct at time $t$ and $t'$} resp.: $(T_p(t) \subseteq T_q(t')) \vee (T_q(t') \subseteq T_p(t))$.
    
    \item[Monotonicity:] For every correct process $p$, $T_p$ can only grow with time: for all $t < t'$, $T_p(t)\subseteq T_p(t')$.
    
    \item[Validity:] If a forever-correct process $p$ invokes $\code{Transfer}(tx, \sigma_{tx})$ at time $t$, then there exists $t'>t$ such that $tx \in T_p(t')$.
    
    \item[Agreement:] For a \atreplace{correct process $p$}{process $p$ correct} at time $t$ and  a forever-correct process $q$, there exists $t'\geq t$ such that $T_p(t)\subseteq T_q(t')$. 
    
\end{description}


Here, $\sigma_{tx} \in \Sigma_{\tT}$ is a certificate for transaction $tx = (p, \tau, D)$. 
A certificate  protects the users from possible theft of their funds by other users.
As we assume that cryptographic techniques (including digital signatures) are unbreakable, the only way to steal someone's funds is to steal their private key.



A natural convention is that a correct process never submits conflicting transactions.
When it comes to Byzantine processes, we make no assumptions.
Our specification ensures that if two conflicting transactions are issued by a Byzantine process, then at most one of them will ever be confirmed. 
In fact, just a single attempt of a process to cheat may lead to the loss of its funds \yareplace{and may}{as it can happen that neither of the conflicting transactions is confirmed and thus may} preclude the process from making progress. 
%


%


\myparagraph{Transaction set as a configuration.}
For simplicity, we assume that the total amount of funds in the system is a publicly known constant $M\in \mathbb{Z^+}$ (fixed by the initial transaction $tx_{init}$) that we call \emph{system stake}.\footnote{In \ppreplace{Section~\ref{sec:discussion}}{Appendix~\ref{app:optimizations}}, we discuss how to maintain a dynamic system stake.} 

A set of transactions $C\in 2^{\tT}$ is called a \emph{configuration}. 

A configuration $C$ is \atreplace{\emph{correct} (or \emph{valid})}{\emph{valid}} if no two transactions $tx_i, tx_j \in C$ are conflicting. 
\atremove{Otherwise, $C$ is \emph{incorrect} (or \emph{invalid}).}
In the rest of this paper, we only consider valid configurations, i.e.,  invalid configurations appearing in protocol messages are ignored by correct processes without being mentioned explicitly in the algorithm description. 

The \emph{initial configuration} is denoted by $C_{init}$, and consists of just one initial transaction ($C_{init} = \{ tx_{init} \}$).

A valid configuration $C$ determines the \emph{stake} (also known as \emph{balance}) of every process $p$ as the difference between the amount of assets sent to $p$ and the amount of assets sent by $p$ in the transactions of $C$: $\textit{stake}(q,C)=\sum_{tx\in C} tx.\tau(q)-\sum_{tx\in C \wedge tx.p=q} \textit{value}(tx)$.
Intuitively, a process joins the asset-transfer system as soon as it gains a positive stake in a configuration and leaves once its stake turns zero.

Functions $\textit{members}$ and $\textit{quorums}$ are defined as follows:
$\textit{members}(C) = \{ p \mid \exists tx \in C \text{ such that } tx.\tau(p) > 0 \}$,
$\textit{quorums}(C) = \{ Q \mid  \sum_{q \in Q} \textit{stake}(q, C) > \frac{2}{3} M \}$.
Intuitively, members of the system in a given configuration $C$ are the processes that 
received money at least once,
and a set of processes is considered to be a quorum in configuration $C$, if their total stake in $C$ is more than two-thirds of the system stake. %

\myparagraph{Configuration lattice.}
Recall that a \emph{join-semilattice} (we simply say a \emph{lattice}) is a tuple $(\Lat, \sqsubseteq)$, where $\Lat$ is a set of \emph{elements} provided with a partial-order relation $\sqsubseteq$, such that for any two elements $a \in \Lat$ and $b \in \Lat$,
there exists the \emph{least upper bound} for the set $\{a, b\}$, i.e., an element $c \in \Lat$ such that $a \sqsubseteq c$, $b \sqsubseteq c$ and $\forall d \in \Lat$: if $a \sqsubseteq d$ and $b \sqsubseteq d$, then $c \sqsubseteq d$.
The least upper bound of elements $a \in \Lat$ and $b \in \Lat$ is denoted by $a \sqcup b$.
\atreplace{$\sqcup$, the associative, commutative and idempotent binary operator on $\Lat$, is called the \emph{join operator}.}{$\sqcup$ is an associative, commutative and idempotent binary operator on $\Lat$. It is called the \emph{join operator}.}

The \emph{configuration lattice} is then defined as $(\cC, \sqsubseteq)$, where $\cC$ is the set of all valid configurations, $\sqsubseteq=\subseteq$ and $\sqcup=\cup$.

\section{{\Name} Asset Transfer: Algorithm}
\label{sec:algo}
In this section we present {\Name} -- an implementation of an asset-transfer system.
We start with the main building blocks of the algorithm, and then proceed to the description of the \Name protocol itself.

We bundle parts of \Name algorithm that are semantically related in building blocks called \textit{objects}, each offering a set of operations, and combine them to implement asset transfer. 

\myparagraph{Transaction Validation.} The \emph{Transaction Validation} (TV) object 
is a part of \Name ensuring that transactions that a correct process $p$ adds to its local state $T_p$ do not conflict.

A correct process $p$ submits a transaction $tx = (p,\tau, D)$  and a matching certificate to the object by invoking an operation $\code{Validate}(tx, \sigma_{tx})$. 
The operation returns a set of transactions $txs \in 2^{\tT}$ together with a certificate $\sigma_{txs} \in \Sigma_{2^{\tT}}$. 
We call a transaction set \textit{verifiable} iff it is verifiable in terms of a function $\code{VerifyTransactionSet}(txs, \sigma_{txs})$. 
%
%
Intuitively, the function returns $\true$ iff certificate $\sigma_{txs}$ confirms that set of transactions $txs$ is \emph{validated} by sufficiently many system members.

Formally, TV satisfies the following properties:
\begin{description}
    \item [TV-Verifiability:] If an invocation of \atreplace{$\code{Validate}(\dots)$}{$\code{Validate}$} returns $\langle txs, \sigma_{txs} \rangle$ to a correct process, then
    $\code{VerifyTransactionSet}(txs, \sigma_{txs}) = \true$;
    \item [TV-Inclusion:]  If $\code{Validate}(tx, \sigma_{tx})$ returns $\langle txs, \sigma_{txs} \rangle$ to a correct process, then $tx \in txs$;
  \item[TV-Validity:] The union of returned verifiable transaction sets consists of non-conflicting verifiable transactions.
\end{description}
In our algorithm, we use one Transaction Validation object $\textit{TxVal}$.

\myparagraph{Adjustable Byzantine Lattice Agreement.} %
Our asset-transfer algorithm reuses elements of an implementation of \emph{Byzantine Lattice Agreement} (BLA), a Lattice Agreement~\citep{gla} protocol that tolerates Byzantine failures.
%
%
We introduce \emph{Adjustable Byzantine Lattice Agreement} (ABLA), an abstraction that captures safety properties of BLA.
An ABLA object is parameterized by a lattice ($\Lat$, $\sqsubseteq$) and a boolean function $\code{VerifyInputValue}: \Lat \times \Sigma_{\Lat} \rightarrow \{\true, \false \}$.
An input value of a given ABLA object is \textit{verifiable} if it complies with   $\code{VerifyInputValue}$.

ABLA exports one operation: $\code{Propose}(v, \sigma_v)$, where $v \in \Lat$ is an input value and $\sigma_v \in \Sigma_{\Lat}$ is a matching certificate.
It also exports function $\code{VerifyOutputValue}: \Lat \times \Sigma_{\Lat} \rightarrow \{\true, \false\}$.
The $\code{Propose}$ operation returns a pair $\langle w, \sigma_w \rangle$, where $w \in \Lat$ is an output value and $\sigma_w \in \Sigma_{\Lat}$ is a matching certificate.
An output value $w$ of an ABLA object is \textit{verifiable} if it complies with function $\code{VerifyOutputValue}$.

An ABLA object satisfies the following properties:
\begin{description}
    \item[ABLA-Validity]: Every verifiable output value $w$ is a join of some set of verifiable input~values;
    
    \item[ABLA-Verifiability:] If an invocation of \atreplace{$\code{Propose}(\dots)$}{$\code{Propose}$} returns $\langle w, \sigma_{w} \rangle$ to a correct process,
    then $\code{VerifyOutputValue}(w, \sigma_w) = \true$;
    
    \item[ABLA-Inclusion:] If $\code{Propose}(v, \sigma_v)$ returns $\langle w, \sigma_w \rangle$, then $v \sqsubseteq w$;
    
    \item[ABLA-Comparability:] All verifiable output values are comparable.
\end{description} 
In our algorithm, we use two ABLA objects, \textit{ConfigLA} and \textit{HistLA}.
%

\sloppy
\myparagraph{Combining TV and ABLA.}
\atreplace{Let us parameterize \textit{ConfigLA} (a configuration lattice) as follows}{Let \textit{ConfigLA} be an ABLA object parameterized as follows}: $(\Lat,\sqsubseteq) = (2^{\tT}, \subseteq)$, \atreplace{$\code{VerifyInputValue}(\dots)$}{$\code{VerifyInputValue}(v, \sigma_v)$} is defined as \atreplace{$\code{VerifyTransactionSet}(\dots)$}{$\code{VerifyTransactionSet}(v, \sigma_v)$}, that is a part of TV. 
The function $\code{VerifyConfiguration}(C, \sigma_C)$ is an alias for $\textit{ConfigLA}.\code{VerifyOutputValue}(C, \sigma_C)$.
Using this object, the processes produce \emph{comparable configurations}, i.e., transaction sets related by containment ($\subseteq$).

\textit{HistLA} is \atadd{an ABLA object} used to produce \emph{sets of configurations} that are all related by containment.
Under the assumption that every input (a set of configurations) to the object only contains comparable configurations, the outputs are related by containment and configurations in these sets are comparable. 
We can see these sets as \emph{sequences} of ordered configurations.
Such sets are called \emph{histories}, as was recently  suggested for asynchronous Byzantine fault-tolerant reconfiguration~\citep{bla}.
\ppreplace{Histories}{It was shown that histories} allow us to access only $O(n)$ configurations, when $n$ configurations are concurrently proposed. 
A history $h$ is called \emph{verifiable} if it complies with $\code{VerifyHistory}(h, \sigma_h)$.

Formally, the ABLA object \textit{HistLA} is parameterized as follows: $(\Lat,\sqsubseteq) = (2^{2^{\tT}}, \subseteq)$.
The requirement that the elements of an input are comparable is established via the function $\code{VerifyInputValue}(\{ \langle C_1, \sigma_1 \rangle, \dots, \langle C_n, \sigma_n \rangle \}) = \bigwedge\limits_{i=1}^n \code{VerifyConfiguration}(C_i, \sigma_i)$. 
%
$\code{VerifyHistory}(H, \sigma_H)$ is an alias for $\textit{HistLA}.\code{VerifyOutputValue}(H, \sigma_H)$.

\begin{figure}[ht]
\centering
\includesvg[width=\textwidth]{LA-PoS/HistoriesProduction}
\caption{\Name pipeline.}
\label{fig:histories_production}
 \vspace{-4mm}
\end{figure}

Thus, the only  valid input values for the \textit{HistLA} object consist of  verifiable output values of the \textit{ConfigLA} object.
At the same time, the only valid input values for the \textit{ConfigLA} object are verifiable transaction sets, returned from the \textit{TxVal} object. 
And, the only valid inputs for a \textit{TxVal} object are signed transactions.
In this case, all verifiable histories are related by containment (comparable via $\subseteq$), and all configurations within one verifiable history are related by containment (comparable via $\subseteq$) as well.
Such a pipeline \atreplace{can help}{helps} us to guarantee that all configurations of the system are comparable and contain non-conflicting transactions. 
Hence, such configurations can act as a consistent representation of stake distribution that changes with time.

To get a high-level idea of how \Name works, imagine a conveyor belt (Figure~\ref{fig:histories_production}). %
Transactions of different users are submitted as inputs, and as outputs, they obtain sets of configurations. 
Then one can choose the configuration representing the largest resulting set and install it in the system\pprev{, changing the funds distribution in the system}. 

\myparagraph{Algorithm overview.}
We assume that blocks of code (functions, operations, callbacks) are executed sequentially until they complete or get interrupted by a wait condition ($\textbf{wait for} \dots$).
Some events, e.g., receiving a message, may trigger callbacks (marked with keyword \textbf{upon}).
However, they are not executed immediately, but firstly placed in an event queue, waiting for their turn. 

By ``$\textbf{let } \textit{var} = \textit{expression}$'' we denote an assignment to a local variable (which can be accessed only from the current function, operation, or callback), and by ``$\textit{var} \leftarrow \textit{expression}$'' we denote an assignment to a global variable (which can be accessed from anywhere by the same process).

%
We denote calls to a weak reliable broadcast primitive with $\textbf{WRB-broadcast}\langle\dots\rangle$ and $\textbf{WRB-deliver}\langle\dots\rangle$.
Besides, we assume a \emph{weak uniform reliable broadcast (WURB)} primitive, a variant of uniform reliable broadcast~\cite{rbroadcast}.
WURB ensures an additional property compared to WRB: if the system remains static (i.e., the configurations stop changing) and a correct process delivers message $m$, then every forever-correct process eventually delivers $m$.  
%
The primitive helps us to ensure that if configuration $C$ is never replaced with a greater one, then every forever-correct process will eventually learn about such a configuration $C$.
To achieve this semantics, before triggering the $\textbf{deliver}$ callback, a correct process just needs to ensure that a quorum of replicas have received the message~\citep{rbroadcast}.
The calls to WURB should be associated with some configuration $C$, and are denoted as $\textbf{WURB-broadcast}\langle\dots,C\rangle$ and $\textbf{WURB-deliver}\langle\dots,C\rangle$. 

The algorithm uses one TV object \textit{TxVal} and two ABLA objects: \textit{ConfigLA} and \textit{HistLA}.
We list the pseudocode for the implementation of the \textit{TxVal} object in Appendix~\ref{app:code}, Figure~\ref{fig:transfervalidation}. 
The implementations of \textit{ConfigLA} and \textit{HistLA} are actually the same, and only differ in their parameters, we therefore provide one generalized implementation for both objects  (Appendix~\ref{app:code}, Figure~\ref{fig:dbla}).

Every process maintains variables $\textit{history}$, $T_p$ and $C_{cur}$, accessible everywhere in the code.
$T_p$ is the latest installed configuration by the process $p$, by $C_{cur}$ we denote the configuration starting from which the process needs to transfer data it stores to greater configurations, and $\textit{history}$ is the current verifiable history (along with its certificate $\sigma_{hist}$).

The main part of \Name protocol (depicted in Figure~\ref{fig:posmain}) exports one operation $\code{Transfer}(tx, \sigma_{tx})$.
%
%
 
%
In this operation, we first set the local variables that store intermediate results produced in this operation, to null values $\perp$. 
Next, the $\code{Validate}$ operation of the Transaction Validation object \textit{TxVal} is invoked with the given transaction $tx$ and the corresponding certificate $\sigma_{tx}$.
The set of transactions $txs$ and certificate $\sigma_{txs}$ returned from the TV object are then used as an input for \code{Propose} operation of $\textit{ConfigLA}$.
Similarly, the result returned from $\textit{ConfigLA}$ ``wrapped in'' a singleton set, is used as an input for \code{Propose} operation of $\textit{HistLA}$. 
The returned verifiable history is then broadcast in the system.
We consider operation $\code{Transfer}(tx, \sigma_{tx})$ to \emph{complete} by a correct process $p$ once process $p$ broadcasts message $\langle \textbf{NewHistory},h, \sigma \rangle$ at \atreplace{line~\ref{line:posmain:historypropose}}{line~\ref{line:posmain:rbhistory}} or if $p$
stops the current $\code{Transfer}$ by executing line~\ref{line:posmain:complete}.
%
Basically, the implementation of \atreplace{$\code{Transfer}(\dots)$}{$\code{Transfer}$} operation follows the logic described before and depicted in Figure~\ref{fig:histories_production}.

\begin{figure}[ht!]
\hrule \vspace{1mm}
 {\footnotesize
\setcounter{mylinenumber}{0}
\begin{tabbing}
 bbbb\=bbbb\=bbbb\=bbbb\=bbbb\=bbbb\=bbbb\=bbbb \=  \kill

$\textbf{typealias}~ \text{SignedTransaction} = \text{Pair}\langle\tT, \Sigma_{\tT}\rangle$ \\
$\textbf{typealias}~ \cC = \text{Set}\langle \tT \rangle$ \\\\
\textbf{Global variables:} \\
$\textit{TxVal} = \text{TV}()$ \\
$\textit{ConfigLA} = \text{ABLA}(\Lat = 2^{\tT}, \sqsubseteq = \subseteq, v_{init} = \atreplace{\langle C_{init}, \perp\rangle}{C_{init}})$ \\
$\textit{ConfigLA}.\text{VerifyInputValue}(v, \sigma) = \text{VerifyTransactionSet}(v, \sigma)$ \\
%
%
$\text{VerifyConfiguration}(v, \sigma) = \textit{ConfigLA}.\text{VerifyOutputValue}(v, \sigma)$ \\
$\textit{HistLA} = \text{ABLA}(\Lat = 2^{2^{\tT}}, \sqsubseteq = \subseteq, v_{init} = \atreplace{\langle \{C_{init}\}, \perp\rangle}{\{C_{init}\}})$\\
$\textit{HistLA}.\text{VerifyInputValue}(\{v\}, \sigma) = \text{VerifyConfiguration}(v, \sigma)$ \\
%
%
$\text{VerifyHistory}(v, \sigma) = \textit{HistLA}.\text{VerifyOutputValue}(v, \sigma)$ \\
\\

$T_p = C_{init}$ \\
$C_{cur}$ = $C_{init}$ \\
\textit{history} = $\{C_{init}\}$, $\sigma_{hist} = \perp $\\\\
$\textit{txs} = \perp$, $\sigma_{txs} = \perp$ \\
$C  = \perp$, $\sigma_C = \perp$ \\
$H = \perp$, $\sigma_H = \perp$ \\
$\textit{curTx} = \perp$ \\
\\

\textbf{operation} Transfer($tx$: $\tT$, $\sigma_{tx}$: $\Sigma_{\tT}$): void \\
\nnll \> $txs \leftarrow \perp$, $\sigma_{txs} \leftarrow \perp$ \\ 
\nnll \> $C \leftarrow \perp$, $\sigma_{C} \leftarrow \perp$ \\
\nnll \> $H \leftarrow \perp$, $\sigma_{H} \leftarrow \perp$ \\
\nnll \> $\textit{curTx} \leftarrow tx$ \\

\nnll \label{line:posmain:validate} \> $txs, \sigma_{txs} \leftarrow \textit{TxVal}.\text{Validate}(tx, \sigma_{tx})$  \\
\nnll  \label{line:posmain:configpropose} \> $C, \sigma_{C} \leftarrow \textit{ConfigLA}.\text{Propose}(txs, \sigma_{txs}) $ \\
\nnll \label{line:posmain:historypropose} \> $H, \sigma_{H} \leftarrow \textit{HistLA}.\text{Propose}(\{C\}, \sigma_{C})$ \\
\nnll \> \atadd{$\textit{curTx} \leftarrow \bot$} \\
\nnll \label{line:posmain:rbhistory} \> \textbf{WRB-broadcast} $\langle \textbf{NewHistory}, H, \sigma_H \rangle$ \\
\\

\textbf{upon WRB-deliver} $\langle \textbf{NewHistory}, h, \sigma\rangle$ from \textit{any}: \\
\nnll \> \textbf{if} VerifyHistory($h$, $\sigma$) \textbf{and} history $\subset h$ \textbf{then} \\
\nnll \>\> \textbf{trigger} event NewHistory(h) \commentline{$\forall C \in h:$ $C$ is a candidate configuration}\\
\nnll \>\> history $\leftarrow h$, $\sigma_{hist} \leftarrow \sigma$ \\
\nnll \>\> \textbf{let} $C_h$ = HighestConf(history) \\
\nnll \>\> UpdateFSKey(height($C_h$)) \\
\nnll \>\> \textbf{if} $curTx \neq \bot$ \textbf{then} \commentline{There is an ongoing \textnormal{\code{Transfer} operation}} \\
\nnll \>\>\> \textbf{if} $curTx \in C_h$ \textbf{then} \commentline{The last issued transaction by $p$ is included in verifiable history } \\
\nnll \label{line:posmain:complete}\>\>\>\> CompleteTransferOperation() \commentline{Stops the ongoing \code{Transfer} operation if any } \\ 
\nnll \label{line:posmain:restart_start} \>\>\> \textbf{else if} $\textit{txs} = \perp$ \textbf{then} \commentline{Operation \code{Transfer} is ongoing and $txs$ has not been received} \\
\nnll \>\>\>\> $\textit{TxVal}.\text{Request}(\emptyset)$ \commentline{Restarts transaction validation by accessing $C_h$ } \\
\nnll \>\>\> \textbf{else if} $C = \perp$ \textbf{then} \commentline{Operation \code{Transfer} is ongoing and $C$ has not been received yet}\\ 
\nnll \>\>\>\> $\textit{ConfigLA}.\text{Refine}(\emptyset) $ \commentline{Restarts verifiable configuration reception by accessing $C_h$ } \\
\nnll \>\>\> \textbf{else if} $H = \perp$ \textbf{then} \commentline{Operation \code{Transfer} is ongoing and $H$ has not been received yet}\\
\nnll \label{line:posmain:restart_end}\>\>\>\> $\textit{HistLA}.\text{Refine}(\emptyset)$ \commentline{Restarts verifiable history reception by accessing $C_h$ }

\end{tabbing}
}
 \hrule
 \caption{\Name: code for process $p$.}
 \label{fig:posmain}
 \vspace{-6mm}
\end{figure}

If a correct process delivers a message $\langle \textbf{NewHistory}, h, \sigma \rangle$, where $\sigma$ is a valid certificate for history $h$ that is greater than its local estimate $\textit{history}$, it ``restarts'' the first step that it has not yet completed in Transfer operation (lines~\ref{line:posmain:restart_start}-\ref{line:posmain:restart_end}). 
For example, if a correct process $p$ receives a message $\langle \textbf{NewHistory}, h, \sigma_h\rangle$, where $\sigma_h$ is a valid certificate for $h$ and $\textit{history} \subset h$ while being in \textit{ConfigLA}, it restarts this step in order to access a greater configuration. 
The result of \textit{ConfigLA} will still be returned to $p$ in the place it has called \atreplace{$\textit{ConfigLA}.\code{Propose}(\dots)$}{$\textit{ConfigLA}.\code{Propose}$}.
Intuitively, we do this in order to reach the most ``up-to-date'' configuration (``up-to-date'' stake holders).

The State Transfer Protocol (Figure~\ref{fig:statetransfer}) helps us to ensure that the properties of the objects (\textit{TxVal}, \textit{ConfigLA} and \textit{HistLA}) are satisfied \emph{across configurations} that our system goes through. 
As system stake is redistributed actively with time, quorums in the system change as well and, hence, we need to pass the data that some quorum knows in configuration $C$ to some quorum of any configuration $C': C \sqsubset C'$, that might be installed after $C$.
The protocol is executed by a correct process after it delivers a verifiable history $h$, such that $C \in h$ and $C_{cur} \sqsubset C$.


\atrev{%
The \textit{height} of a configuration $C$ is the number of transactions in it (denoted as $\textit{height}(C) = |C|$).
Since all configurations installed in {\Name} are comparable, height can be used as a unique identifier for an installed configuration.}
We \atadd{also} use height as the \emph{timestamp} for forward-secure digital signatures.  
\atremove{A configuration determines the ``current'' stake of each participant.}
When a process $p$ answers requests in configuration $C$, it signs the message with timestamp $\textit{height}(C)$\atremove{, i.e., $p$ ``proves'' that it currently sees a stake distribution defined by configuration $C$}.
%
The process $p$ invokes $\text{UpdateFSKey}(\textit{height}(C'))$ when \atreplace{a new configuration $C' : C \sqsubset C'$ is installed}{it discovers a new configuration $C' : C \sqsubset C'$}. 
Thus, processes that still consider $C$ as the current configuration (and see the corresponding stake distribution) cannot be deceived by a process $p$, that was correct in $C$, but not in a higher installed configuration $C'$ (e.g., $p$  spent all its stake by submitting transactions, \atreplace{that}{which} became part of configuration $C'$, thereby lost its weight in the system, and later became Byzantine).

The implementation of verifying functions \atreplace{(Figure~\ref{fig:transfervalidation})}{(Figure~\ref{fig:verifyingfuncs})} and  a description of the auxiliary functions used in the pseudocode are delegated to Appendix~\ref{app:code}.

\begin{figure}[ht]
\hrule \vspace{1mm}
 {\footnotesize
\begin{tabbing}
 bbbb\=bbbb\=bbbb\=bbbb\=bbbb\=bbbb\=bbbb\=bbbb \=  \kill

\textbf{upon} $C_{cur} \not = \text{HighestConf}(\textit{history})$: \\
\nnll \> \textbf{let} $C_{next}$ = HighestConf(\textit{history}) \\ 
\nnll \> \textbf{let} $S = \{ C \in \textit{history} \mid C_{cur} \sqsubseteq C \sqsubset C_{next} \}$ \\
\nnll \> $\textit{seqNum} \leftarrow \textit{seqNum} + 1$ \\
\nnll \> \textbf{for} $C$ \textbf{in} $S$: \\
\nnll \>\>  \textbf{send} $\langle\textbf{UpdateRead}, \textit{seqNum},C\rangle$ to $\textit{members}(C)$\\
\nnll\label{line:statetrans:wait}\>\>\textbf{wait for} $(C \sqsubset C_{cur}~\textbf{or}~\exists Q \in \textit{quorums}(C)~\text{responded with \textit{seqNum}})$ \\
\nnll \> \textbf{if} $C_{cur} \sqsubset C_{next}$ \textbf{then} \\
\nnll \>\> $C_{cur} \leftarrow C_{next}$ \\
\nnll \>\> \textbf{WURB-broadcast} $\langle \textbf{UpdateComplete},C_{next} \rangle$\\
\\

\textbf{upon} receive  $\langle \textbf{UpdateRead}, sn, C\rangle$ from \textit{q}: \\
\nnll \> \textbf{wait for} $C \sqsubset \text{HighestConf}(\textit{history})$ \\
\nnll \> \textbf{let} $\textit{txs} = \textit{TxVal.seenTxs}$ \\
\nnll \> \textbf{let} $\textit{values}_1 = \textit{ConfigLA.values}$ \\
\nnll \> \textbf{let} $\textit{values}_2 = \textit{HistLA.values}$ \\
\nnll \> \textbf{send} $\langle \textbf{UpdateReadResp}, \textit{txs}, \textit{values}_1, \textit{values}_2, sn \rangle$ to \textit{q} \\ 
\\

\textbf{upon} receive  $\langle \textbf{UpdateReadResp}, \textit{txs}, \textit{values}_1, \textit{values}_2, sn\rangle$ from \textit{q}: \\
\nnll \> \textbf{if} VerifySenders($txs$) \textbf{and} \textit{ConfigLA}.VerifyValues($\textit{values}_1$) \\ \nnll \>\>\> \textbf{and} \textit{HistLA}.VerifyValues($\textit{values}_2$) \textbf{then} \\
\nnll \>\> $\textit{TxVal.seenTxs} \leftarrow \textit{TxVal.seenTxs} \cup txs$ \\
\nnll \>\> $\textit{ConfigLA.values} \leftarrow \textit{ConfigLA.values} \cup \textit{values}_1.\text{filter}(\langle v, \sigma_v\rangle \Rightarrow v \not\in \textit{ConfigLA.values}.\text{firsts}())$ \\
\nnll \>\> $\textit{HistLA.values} \leftarrow \textit{HistLA.values} \cup \textit{values}_2.\text{filter}(\langle v, \sigma_v\rangle \Rightarrow v \not\in \textit{HistLA.values}.\text{firsts}())$ \\
\\

\textbf{upon} \textbf{WURB-deliver} $\langle \textbf{UpdateComplete}, C \rangle$ from quorum $Q \in \textit{quorums}(C)$: \\
\nnll \> \textbf{wait for} $C \in \textit{history}$ \\ 
\nnll \> \textbf{if} $T_p \sqsubset C$ \textbf{then} \\
\nnll \>\> \textbf{if} $C_{cur} \sqsubset C$ \textbf{then} $C_{cur} \leftarrow C$ \\
\nnll \label{line:statetransfer:updateinstalledconfig} \>\> $T_p \leftarrow C$ \commentline{ Update set of ``confirmed'' transactions }\\
\nnll \>\> \textbf{trigger} event InstalledConfig($C$) \commentline{$C$ is an installed configuration}

\end{tabbing}
}
 \hrule
 \caption{State Transfer Protocol: code for process $p$.}
 \label{fig:statetransfer}
 \vspace{-4mm}
\end{figure}

\myparagraph{Implementing Transaction Validation.}
The implementation of the TV object \textit{TxVal} is depicted in Figure~\ref{fig:transfervalidation} \atadd{(Appendix~\ref{app:code})}.
The algorithm can be divided into two phases.

In the first phase, process $p$ sends a message request that contains a set of transactions $\textit{sentTxs}$ to be validated to all members of the current configuration.
Every correct process $q$ that receives such messages first checks whether the transactions have actually been issued by their senders.
If yes, $q$ adds the transactions in the message to the set of transactions it has seen so far and checks whether any transaction from the ones it has just received conflicts with some other transaction it knows about. 
All conflicting transactions are placed in set \textit{conflictTxs}. After $q$ validates transactions, it sends a message $\langle \textbf{ValidateResp}, \textit{txs}, \textit{conflictTxs}, sig, sn\rangle$. 
Here, $txs$ is the union of verifiable transactions received by $p$ from $q$ just now and all other non-conflicting transactions $p$ is aware of so far.
Process $p$ then verifies a received message $\langle \textbf{ValidateResp}, \textit{txs}_q,\textit{conflictTxs}_q, sig_q, sn\rangle$. 
The message received from $q$ is considered to be valid by $p$ if $q$ has signed it with a private key that corresponds to the current configuration, all the transactions from $\textit{txs}_q$ have valid certificates, and if for any verifiable transaction $tx$ from $\textit{conflictTxs}$ there is a verifiable transaction $tx$ also from $\textit{conflictTxs}$, such that $tx$ conflicts with $tx'$.
If the received message is valid and $\textit{txs}_q$ equals $\textit{sentTxs}$, then $p$ adds signature of process $q$ and its validation result to its local set $\textit{acks}_1$.
In case $\textit{sentTxs} \subset \textit{txs}_q$, the whole phase is restarted.
The first phase is considered to be completed as soon as $p$ collects responses from some quorum of processes in $\textit{acks}_1$. 

Such implementation of the first phase makes correct process $p$ obtain certificate not only for its transaction, but also for other non-conflicting transactions issued by other processes.
This helping mechanism ensures that transactions of forever-correct processes are eventually confirmed and become part of some verifiable configuration. 

In the second phase, $p$ collects signatures from a weighted quorum of the current configuration.
If $p$ successfully collects such a set of signatures, then the configuration it saw during the first phase was \emph{active} (no greater configuration had been installed) and it is safe for $p$ to return the obtained result.
This way the  so-called ``slow reader'' attack~\cite{bla} is anticipated.
%
%

If during any of the two phases $p$ receives a message with a new verifiable history that is greater (w.r.t. $\subseteq$) than its local estimate and does not contain last issued transaction by $p$, the described algorithm starts over.
We guarantee that the number of restarts in \textit{TxVal} in \Name protocol is finite only for \textit{forever-correct} processes (please refer to Appendix~\ref{app:proof} for a detailed proof). 
Note that we cannot guarantee this for \textit{all} correct processes as during the protocol execution some of them can become Byzantine.

In the implementation, we assume that the dependency set of a transaction only includes transactions that are confirmed (i.e., included in some installed configuration), otherwise they are considered invalid.

\myparagraph{Implementing Adjustable Byzantine Lattice Agreement.} 
The generalized implementation of ABLA objects \textit{ConfigLA} and \textit{HistLA} is specified by Figure~\ref{fig:dbla} (Appendix~\ref{app:code}).
The algorithm is generally inspired by the implementation of Dynamic Byzantine Lattice Agreement from~\cite{bla}, but there are a few major differences.
Most importantly, it is tailored to work even if the number of reconfiguration requests (i.e., transactions) is infinite.
Similarly to the Transaction Validation implementation, algorithm consists of two phases.

In the first phase, process $p$ sends a message that contains the verifiable inputs it knows to other members and then waits for a weighted quorum of processes of the current configuration to respond with the same set. If $p$ receives a message with a greater set (w.r.t. $\subseteq$), it restarts the phase.
The validation process performed by the processes is very similar to the one used in Transaction Validation object implementation.

The second phase, in fact, is identical to the second phase of the Transaction Validation implementation.
We describe it in the pseudocode for completeness.

As with Transaction Validation, whenever $p$ delivers a verifiable history $h$, such that it is greater than its own local estimate and $h$ does not contain last issued transaction by $p$, the described algorithm starts over. Similarly to \textit{TxVal}, it is guaranteed that the number of restarts a forever-correct process make in both \textit{ConfigLA} and \textit{HistLA} is finite.

\section{Proof of Correctness}
\label{sec:proof}
In this section, we outline the proof that \Name indeed satisfies the properties of an asset-transfer system.
First, we formulate a restriction we impose on the adversary that is required for our implementation to be correct. 
Informally, the adversary is not allowed to corrupt one third or more stake in a ``candidate'' configuration, i.e., in a configuration that can potentially be used for adding new transactions. 
The adversary is free to corrupt a configuration as soon as it is superseded by a strictly higher one.
We then sketch the main arguments of our correctness proof (the detailed proof is deferred to Appendix~\ref{app:proof}). 

\subsection{Adversarial restrictions}
\label{sec:adversary}

A configuration $C$ is considered to be \emph{installed} if some correct process has triggered the special event $\text{InstalledConfig}(C)$.
We call a configuration $C$ a \emph{candidate} configuration if some correct process has triggered a $\text{NewHistory}(h)$ event, such that $C \in h$.
We also say that a configuration $C$ is \emph{superseded} if some correct process installed a higher configuration $C'$. 
An installed (resp., candidate) configuration $C$ is called an \emph{active} (resp., an \emph{active candidate}) configuration as long as it is not superseded.
Note that at any moment of time $t$, every active installed configuration is an active candidate configuration, but not vice versa.

%
%

We expect the adversary to obey the following condition:

\begin{description}

\item[Configuration availability:] 
Let $C$ be an active candidate configuration at time $t$ and let $\textit{correct}(C,t)$ denote the set of processes in $\textit{members}(C)$ that are correct at time $t$.   Then $C$ must be \emph{available} at time $t$: 
\[
\sum_{q\in\textit{correct}(C,t)}\textit{stake}(q,C) >2/3 M.
\]
\end{description}

Note that the condition allows the adversary to compromise a candidate configuration once it is superseded by a more recent one. 
As we shall see, the condition implies that our algorithm is live. 
Intuitively, a process with a pending operation will either eventually hear from the members of a configuration holding ``enough'' stake which might allow it to complete its operation or will learn about a more recent configuration, in which case it can abandon the superseded configuration and proceed to the new one.

\subsection{Proof outline}

%

\myparagraph{Consistency}. 
The consistency property states that (1)~as long as process $p$ is correct, $T_p$ contains only verifiable non-conflicting transactions and (2)~if processes $p$ and $q$ are correct at times $t$ and $t'$ respectively, then $T_p(t) \subseteq T_q(t')$ or $T_q(t') \subseteq T_p(t)$.
To prove that \Name satisfies this property, we show that our implementation of \textit{TxVal} meets the specification of Transaction Validation and that both \textit{ConfigLA} and \textit{HistLA} objects satisfy the properties of Adjustable Byzantine Lattice Agreement.
Correctness of \textit{HistLA} ensures that all verifiable histories are related by containment and the correctness of \textit{ConfigLA} guarantees that all verifiable configurations are related by containment (i.e., they are \emph{comparable}).
Taking into account that the only possible verifiable inputs for \textit{HistLA} are sets that contain verifiable output values (configurations) of \textit{ConfigLA}, we obtain the fact that all configurations of any verifiable history are comparable as well.
As all installed configurations (all $C$ such that a correct process triggers an event $\text{InstalledConfig}(C)$) are elements of some verifiable history, they all are related by containment too.
Since $T_p$ is in fact the last configuration installed by a process $p$, we obtain (2).
The fact that every $p$ stores verifiable non-conflicting transactions follows from the fact that the only possible verifiable input values for \textit{ConfigLA} are the output transaction sets returned by \textit{TxVal}. As \textit{TxVal} is a correct implementation of Transaction Validation, then union of all such sets contain verifiable non-conflicting transactions. Hence, the only verifiable configurations that are produced by the algorithm cannot contain conflicting transactions. From this we obtain (1).

\myparagraph{Monotonicity.} This property requires that $T_p$ only grows as long as $p$ is correct. The monotonicity of \Name follows from the fact that correct processes install only greater configurations with time, and that the last installed configuration by a correct process $p$ is exactly $T_p$. 
Thus, if $p$ is correct at time $t'$, then for all $t < t'$: $T_p(t) \subseteq T_p(t')$.

\myparagraph{Validity.} This property requires that a transfer operation for a transaction $tx$ initiated by a forever-correct process will lead to $tx$ being included in $T_p$ at some point in time. In order to prove this property for \Name, we show that a forever-correct process $p$ may only be blocked in the execution of a \code{Transfer} operation if some other process successfully installed a new configuration.
We argue that from some moment of time on every other process that succeeds in the installation of configuration $C$ will include the transaction issued by $p$ in the $C$.
We also show that if such a configuration $C$ is installed then eventually every forever-correct process installs a configuration $C': C \sqsubseteq C'$.
As $T_p$ is exactly the last configuration installed by $p$, eventually any transaction issued by a forever-correct process $p$ is included in $T_p$.

\myparagraph{Agreement.} In the end we show that the \Name protocol satisfies the agreement property of an asset-transfer system. 
The property states the following: for a correct process $p$ at time $t$ and a forever-correct process $q$, there exists $t' \ge t$ such that $T_p(t) \subseteq T_p(t')$. Basically, it guarantees that if a transaction $tx$ was considered confirmed by $p$ when it was correct, then any forever-correct process will eventually confirm it as well.
To prove this, we show that if a configuration $C$ is installed by a correct process, then every other forever-correct process will install some configuration $C'$, such that $C \sqsubseteq C'$.
Taking into account the fact that $T_p$ is a last configuration installed by a process $p$, we obtain the desired.

\section{Concluding Remarks}
\label{sec:discussion}

\Name is a permissionless asset transfer system based on proof of stake. It builds on lattice agreement primitives and provides its guarantees in asynchronous environments where less than one third of the total stake is owned by malicious parties.
%

\myparagraph{Enhancements and optimizations.}
To keep the presentation focused, so far we described only the core of the \Name protocol.
However, there is a number of challenges that need to be addressed before the protocol can be applied in practice. 
\yareplace{%
We discuss some of them in Appendix~\ref{app:optimizations}.
First of all, we analyze the communication and storage complexity of the protocol and suggest a number of optimizations.
Then we suggest possible implementations for \emph{delegation}, \emph{fees} and \emph{inflation} mechanisms.
Note that these are non-trivial to implement in an asynchronous system because there is no clear agreement on the order in which transactions are added to the system or on the distribution of stake at the moment when a transaction is added to the system.
Finally, we discuss some practical aspects of using forward-secure digital signature schemes.}%
{In particular the communication, computation and storage complexity of the protocol can be improved with carefully constructed messages and signature schemes\atreplace{ and}{. Also,} mechanisms for \atreplace{the separation of processes into clients and replicas}{stake delegation} as well as fees and inflation are necessary. Note that these are non-trivial to implement in an asynchronous system because there is no clear agreement on the order in which transactions are added to the system or on the distribution of stake at the moment when a transaction is added to the system.
We discuss these topics \atadd{as well as the practical aspects of using forward-secure signatures} in Appendix~\ref{app:optimizations}\atremove{ in more detail}.}

\myparagraph{Open questions.}
%
%
In this paper, we demonstrated that it is possible to combine asynchronous cryptocurrencies with proof-of-stake in presence of a dynamic adversary.
However, there are still plenty of open questions.
Perhaps, the most important direction is to study hybrid solutions which combine our approach with consensus in an efficient way in order to support general-purpose smart contracts.
Further research is also needed in order to improve the efficiency of the solution and to measure how well it will behave in practice compared to consensus-based solutions.
Finally, designing proper mechanisms in order to incentivize active and honest participation is a non-trivial problem in the harsh world of asynchrony, where the processes cannot agree on a total order in which transactions are executed.

\bibliographystyle{abbrv}
\bibliography{references}  

\begin{thebibliography}{10}

\bibitem{bellare1999forward}
M.~Bellare and S.~K. Miner.
\newblock A forward-secure digital signature scheme.
\newblock In {\em Annual International Cryptology Conference}, pages 431--448,
  Berlin, 1999. Springer.

\bibitem{PoS16}
I.~Bentov, A.~Gabizon, and A.~Mizrahi.
\newblock Cryptocurrencies without proof of work.
\newblock In {\em Financial Cryptography and Data Security - {FC} 2016
  International Workshops, BITCOIN, VOTING, and WAHC, Christ Church, Barbados,
  February 26, 2016, Revised Selected Papers}, pages 142--157, Berlin, 2016.
  Springer.

\bibitem{rbroadcast}
C.~Cachin, R.~Guerraoui, and L.~Rodrigues.
\newblock {\em Introduction to Reliable and Secure Distributed Programming}.
\newblock Springer Publishing Company, Incorporated, 2nd edition, 2011.

\bibitem{algorand}
J.~Chen and S.~Micali.
\newblock Algorand: {A} secure and efficient distributed ledger.
\newblock {\em Theor. Comput. Sci.}, 777:155--183, 2019.

\bibitem{coinmarketcap}
CoinMarketCap.
\newblock Cryptocurrency prices, charts and market capitalizations, 2021.
\newblock \url{https://coinmarketcap.com/}, accessed 2021-02-15.

\bibitem{astro-dsn}
D.~Collins, R.~Guerraoui, J.~Komatovic, P.~Kuznetsov, M.~Monti, M.~Pavlovic,
  Y.~A. Pignolet, D.~Seredinschi, A.~Tonkikh, and A.~Xygkis.
\newblock Online payments by merely broadcasting messages.
\newblock In {\em 50th Annual {IEEE/IFIP} International Conference on
  Dependable Systems and Networks, {DSN} 2020, Valencia, Spain, June 29 - July
  2, 2020}, pages 26--38. {IEEE}, 2020.

\bibitem{sybil}
J.~R. Douceur.
\newblock The sybil attack.
\newblock In {\em Peer-to-Peer Systems, First International Workshop, {IPTPS}
  2002, Cambridge, MA, USA, March 7-8, 2002, Revised Papers}, pages 251--260,
  Heidelberg, 2002. Springer.

\bibitem{drijvers2019pixel}
M.~Drijvers, S.~Gorbunov, G.~Neven, and H.~Wee.
\newblock Pixel: Multi-signatures for consensus.
\newblock In {\em 29th {USENIX} Security Symposium ({USENIX} Security 20)},
  Boston, MA, Aug. 2020. {USENIX} Association.

\bibitem{PoSE15}
S.~Dziembowski, S.~Faust, V.~Kolmogorov, and K.~Pietrzak.
\newblock Proofs of space.
\newblock In {\em Advances in Cryptology - {CRYPTO} 2015 - 35th Annual
  Cryptology Conference, Santa Barbara, CA, USA, August 16-20, 2015,
  Proceedings, Part {II}}, pages 585--605, 2015.

\bibitem{gla}
J.~M. Falerio, S.~K. Rajamani, K.~Rajan, G.~Ramalingam, and K.~Vaswani.
\newblock Generalized lattice agreement.
\newblock In D.~Kowalski and A.~Panconesi, editors, {\em {ACM} Symposium on
  Principles of Distributed Computing, {PODC} '12, Funchal, Madeira, Portugal,
  July 16-18, 2012}, pages 125--134. {ACM}, 2012.

\bibitem{FLP85}
M.~J. Fischer, N.~A. Lynch, and M.~S. Paterson.
\newblock Impossibility of distributed consensus with one faulty process.
\newblock {\em JACM}, 32(2):374--382, Apr. 1985.

\bibitem{quorums}
D.~K. Gifford.
\newblock Weighted voting for replicated data.
\newblock In {\em SOSP}, pages 150--162, 1979.

\bibitem{gilad2017algorand}
Y.~Gilad, R.~Hemo, S.~Micali, G.~Vlachos, and N.~Zeldovich.
\newblock Algorand: Scaling byzantine agreements for cryptocurrencies.
\newblock In {\em Proceedings of the 26th Symposium on Operating Systems
  Principles}, pages 51--68, 2017.

\bibitem{byz-bcast}
R.~Guerraoui, J.~Komatovic, P.~Kuznetsov, Y.~A. Pignolet, D.~Seredinschi, and
  A.~Tonkikh.
\newblock Dynamic {Byzantine} reliable broadcast.
\newblock In {\em OPODIS}, 2020.

\bibitem{at2-cons}
R.~Guerraoui, P.~Kuznetsov, M.~Monti, M.~Pavlovic, and D.-A. Seredinschi.
\newblock The consensus number of a cryptocurrency.
\newblock In {\em PODC}, 2019.
\newblock \url{https://arxiv.org/abs/1906.05574}.

\bibitem{Gup16}
S.~Gupta.
\newblock {A Non-Consensus Based Decentralized Financial Transaction Processing
  Model with Support for Efficient Auditing}.
\newblock Master's thesis, Arizona State University, USA, 2016.

\bibitem{Gossiping}
A.-M. Kermarrec and M.~van Steen.
\newblock Gossiping in distributed systems.
\newblock {\em SIGOPS Oper. Syst. Rev.}, 41(5):2–7, Oct. 2007.

\bibitem{ouroboros17}
A.~Kiayias, A.~Russell, B.~David, and R.~Oliynykov.
\newblock Ouroboros: {A} provably secure proof-of-stake blockchain protocol.
\newblock In {\em Advances in Cryptology - {CRYPTO} 2017 - 37th Annual
  International Cryptology Conference, Santa Barbara, CA, USA, August 20-24,
  2017, Proceedings, Part {I}}, pages 357--388, 2017.

\bibitem{rla}
P.~Kuznetsov, T.~Rieutord, and S.~{Tucci Piergiovanni}.
\newblock Reconfigurable lattice agreement and applications.
\newblock In P.~Felber, R.~Friedman, S.~Gilbert, and A.~Miller, editors, {\em
  23rd International Conference on Principles of Distributed Systems, {OPODIS}
  2019, December 17-19, 2019, Neuch{\^{a}}tel, Switzerland}, volume 153 of {\em
  LIPIcs}, pages 31:1--31:17, 2019.

\bibitem{bla}
P.~Kuznetsov and A.~Tonkikh.
\newblock Asynchronous reconfiguration with byzantine failures.
\newblock In H.~Attiya, editor, {\em 34th International Symposium on
  Distributed Computing, {DISC} 2020, October 12-16, 2020, Virtual Conference},
  volume 179 of {\em LIPIcs}, pages 27:1--27:17, 2020.

\bibitem{byz-quorums}
D.~Malkhi and M.~Reiter.
\newblock Byzantine quorum systems.
\newblock {\em Distributed Computing}, 11?(4):203--213, 1998.

\bibitem{MaMiMi02}
T.~Malkin, D.~Micciancio, and S.~Miner.
\newblock Efficient generic forward-secure signatures with an unbounded number
  of time periods.
\newblock In {\em Advances in Cryptology - Eurocrypt 2002}, volume 2332 of {\em
  Lecture Notes in Computer Science}, pages 400--417, Amsterdam, The
  Netherlands, April 28-May 2 2002. IACR, Springer-Verlag.

\bibitem{PoST16}
T.~Moran and I.~Orlov.
\newblock Proofs of space-time and rational proofs of storage.
\newblock {\em {IACR} Cryptology ePrint Archive}, 2016:35, 2016.

\bibitem{bitcoin}
S.~Nakamoto.
\newblock Bitcoin: A peer-to-peer electronic cash system, 2008.

\bibitem{common-coin}
M.~O. Rabin.
\newblock Randomized byzantine generals.
\newblock In {\em 24th Annual Symposium on Foundations of Computer Science
  (sfcs 1983)}, pages 403--409. IEEE, 1983.

\bibitem{crdt}
M.~Shapiro, N.~M. Pregui{\c{c}}a, C.~Baquero, and M.~Zawirski.
\newblock Conflict-free replicated data types.
\newblock In {\em SSS}, pages 386--400, 2011.

\bibitem{abc-tr}
J.~Sliwinski and R.~Wattenhofer.
\newblock {ABC:} asynchronous blockchain without consensus.
\newblock {\em CoRR}, abs/1909.10926, 2019.

\bibitem{SKM17-reconf}
A.~Spiegelman, I.~Keidar, and D.~Malkhi.
\newblock Dynamic reconfiguration: Abstraction and optimal asynchronous
  solution.
\newblock In {\em DISC}, pages 40:1--40:15, 2017.

\bibitem{ethereum_pos}
P.~Wackerow, R.~Cordell, Tentodev, A.~Stockinger, and S.~Richards.
\newblock Ethereum proof-of-stake (pos), 2008.
\newblock \url{
  https://ethereum.org/en/developers/docs/consensus-mechanisms/pos/}, accessed
  2021-02-15.

\bibitem{ethereum}
G.~Wood.
\newblock Ethereum: A secure decentralized generalized transaction ledger.
\newblock White paper, 2015.

\end{thebibliography}

\newpage
\appendix

\section{Pseudocode}
\label{app:code}

\myparagraph{Verifying and auxiliary functions.}
We provide implementations of required verifying functions used in \textit{ConfigLA} and \textit{HistLA} in Figure~\ref{fig:verifyingfuncs}.
The implementation of function $\code{VerifySender}(tx, \sigma_{tx})$ is not presented there, as it simply consists in verifying that $\sigma_{tx}$ is a valid signature for $tx$ under $tx.p$'s public key.

We use the following auxiliary functions in the pseudocode to keep it concise: 
\begin{itemize}
    \item $\code{VerifySenders}(txs)$ -- returns \textit{true}  iff $\forall \langle tx, \sigma_{tx} \rangle \in txs:$ $\code{VerifySender}(tx, \sigma_{tx}) = \true$, otherwise returns $\false$;
    \item $\code{VerifyValues}(vs)$ -- returns \textit{true} if $\forall \langle v, \sigma \rangle \in vs:$ $\code{VerifyInputValue}(v, \sigma) = \true$, otherwise  returns $\false$; 
    \item $\code{ContainsQuorum}(acks, C)$ -- returns \textit{true} if $\exists Q \in \textit{quorums}(C)$ such that $\forall q \in Q \langle q, \dots \rangle \in acks$, otherwise returns $\false$; 
    \item $\code{HighestConf}(h)$ -- returns the highest (w.r.t. $\sqsubseteq$) configuration in given history $h$;
    \item $\code{firsts}()$ -- method that, when invoked on a set of tuples $S$, returns a set of first elements of tuples from set $S$;
    \item $\code{ConflictTransactions}(txs)$ -- for a given set of verifiable transactions $txs$ returns a set $\textit{conflictTxs}$ such that $\textit{conflictTxs} \subseteq txs$ and for any $\langle tx, \sigma_{tx}\rangle \in \textit{conflictTxs}$ there exists $\langle tx', \sigma_{tx'} \rangle \in \textit{conflictTxs}$ such that $tx$ conflicts with $tx'$; 
    \item $\code{CorrectTransactions}(txs)$ -- returns a set of transactions $\textit{correctTxs}$ such that $\textit{correctTxs} \subseteq txs$, $\langle tx,\sigma_{tx} \rangle \in \textit{correctTxs}$ iff $\langle tx, \sigma_{tx} \rangle \in txs$, $\sigma_{tx}$ is a valid certificate for $tx$ and $\nexists \langle  tx', \sigma_{tx'} \rangle \in txs$, such  that $tx$ conflicts with $tx'$.
\end{itemize}

\begin{figure}[ht]
\hrule \vspace{1mm}
 {\small
\begin{tabbing}
 bbbb\=bbbb\=bbbb\=bbbb\=bbbb\=bbbb\=bbbb\=bbbb \=  \kill

\textbf{fun} $\text{VerifyTransactionSet}(txs: \text{Set}\langle \tT \rangle, \sigma_{txs}: \Sigma_{2^{\tT}}): \text{Bool}$\\
\nnll \> $\textbf{let}~\langle \textit{sentTxs}, \textit{acks}_1, \textit{acks}_2, \textit{h},\sigma_h \rangle = \sigma_{txs}$ \\  
\nnll \> $\textbf{let}~C = \text{HighestConf}(h)$ \\
\nnll \> $\textbf{return}~\text{VerifyHistory}(h, \sigma_h) \textbf{ and } txs = \textit{sentTxs}.\text{firsts}() \setminus \textit{acks}_1.\text{getConflictTxs}().\text{firsts}()$  \\ 
\nnll \>\> \textbf{and } $\text{ContainsQuorum}(\textit{acks}_1, C)\textbf{ and } \text{ContainsQuorum}(\textit{acks}_2, C)$ \\ 
\nnll \>\> \textbf{and } $\textit{acks}_1.\text{forAll}(\langle q, \textit{sig}, \textit{conflictTxs}\rangle \Rightarrow$ \\
\nnll \>\>\> $\text{FSVerify}(\langle \textbf{ValidateResp}, \textit{sentTxs}, \textit{conflictTxs} \rangle, q, \textit{sig}, \textit{height}(C)))$ \\
\nnll \>\> \textbf{and } $\textit{acks}_1.\text{forAll}(\langle q, \textit{sig}, \textit{conflictTxs}\rangle \Rightarrow $ \\
\nnll \>\>\> $\textit{conflictTxs}.\text{forAll}(\langle tx, \sigma_{tx} \rangle \Rightarrow tx \text{ conflicts with } tx' \text{ such that } \langle tx', \sigma_{tx'} \rangle \in \textit{conflictTxs}))$ \\
\nnll \>\> \textbf{and } $\textit{acks}_2.\text{forAll}(\langle q, \textit{sig} \rangle \Rightarrow \text{FSVerify}(\langle \textbf{ConfirmResp}, \textit{acks}_1\rangle, q,  \textit{sig}, \textit{height}(C)))$ \\
\\

\textbf{fun} $\text{ABLA}.\text{VerifyOutputValue}(v: \Lat, \sigma: \Sigma): \text{Bool}$\\
\nnll \> $\textbf{if }\sigma = \perp \textbf{ then return } v = v_{init}$ \\
\nnll \> $\textbf{let}~\langle \textit{values}, \textit{acks}_1, \textit{acks}_2, \textit{h},\sigma_h \rangle = \sigma$ \\  
\nnll \> $\textbf{let}~C = \text{HighestConf}(h)$ \\
\nnll \> $\textbf{return}~\text{VerifyHistory}(h, \sigma_h) \textbf{ and } v = \bigsqcup \textit{values}.\text{firsts}()$  \\ 
\nnll \>\> \textbf{and } $\text{ContainsQuorum}(\textit{acks}_1, C)\textbf{ and } \text{ContainsQuorum}(\textit{acks}_2, C)$ \\ 
\nnll \>\> \textbf{and } $\textit{acks}_1.\text{forAll}(\langle q, \textit{sig} \rangle \Rightarrow \text{FSVerify}(\langle \textbf{ProposeResp}, \textit{values} \rangle, q, \textit{sig}, \textit{height}(C)))$ \\
\nnll \>\> \textbf{and } $\textit{acks}_2.\text{forAll}(\langle q, \textit{sig} \rangle \Rightarrow \text{FSVerify}(\langle \textbf{ConfirmResp}, \textit{acks}_1\rangle, q, \textit{sig}, \textit{height}(C)))$ \\
\\

\textbf{fun }$\text{VerifyConfiguration}(v, \sigma) = \textit{ConfigLA}.\text{VerifyOutputValue}(v, \sigma)$ \\
\\

\textbf{fun }$\text{VerifyHistory}(v, \sigma) = \textit{HistLA}.\text{VerifyOutputValue}(v, \sigma)$

\end{tabbing}
}
 \hrule
 \caption{Verifying functions: code for process $p$.}
 \label{fig:verifyingfuncs}
\end{figure}

\begin{figure}[H]
\input{LA-PoS/CodeTxVal}
\end{figure}

\begin{figure}[H]
\input{LA-PoS/CodeABLA}
\end{figure}



\section{Proof of Correctness}
\label{app:proof}

\myparagraph{Consistency}. We start with the proof of the consistency property of \Name. 
We show that all verifiable histories produced by the algorithm are related by containment and all configurations within one such history are comparable.
Note that in our algorithm, verifiable histories are not only the ``results'', but they  also act as ``the state'' on which the future result might depend.
(Recall that if a history $h$ is verifiable, then some correct process triggered an event $\text{NewHistory}(h)$.)

Intuitively, we proceed by induction on the length of verifiable histories.
We prove correctness of \textit{TxVal}, \textit{ConfigLA} and \textit{HistLA} under the assumption that all verifiable histories are related by containment and all configurations within one verifiable history are comparable,
as though these verifiable histories are obtained from some trusted source. 
Then, after correctness of the used objects is proved, we show that all verifiable histories that the protocol produces are related by containment and all configurations within one such history are comparable. 

The assumption mentioned above is proved to hold by induction.
We assume that the verifiable histories, that are used as the ``state'' of algorithm satisfy the mentioned properties. 
Under this assumption we then prove that the obtained results also satisfy the required properties. 
%
The induction base is trivial: all histories within a set of verifiable histories that consists only of initial history $h_{init} = \{ C_{init} \}$ are related by containment and all configurations within one such history are comparable.
The induction hypothesis states that we can fix any finite prefix of a distributed execution, and assume that all verifiable histories $hs$ (more specifically, the ones for which event $\text{NewHistory}(h)$ was triggered) are related by containment and contain only comparable configurations.
Now we need to prove the induction step. In order to do this we prove that the ``next'' produced verifiable histories by the algorithm are also related by containment between themselves and any $h \in hs$ and contain only comparable configurations.

Recall that a configuration $C$ is a \emph{candidate} configuration iff a correct process triggers an event $\text{NewHistory}(h)$, such that $C \in h$. 
From the definition and our implementation of \Name it follows that if a configuration is a \emph{candidate} configuration, then it appears in some verifiable history.

\begin{lemma}
\label{lm:candidateconfigurations}
All candidate configurations are comparable with $\sqsubseteq$.
\end{lemma}
\begin{proof}
All verifiable histories are required to be related by containment and all configurations within one history are required to be comparable as well.
\end{proof}

We say that a configuration is \emph{pivotal} if it is the last configuration in some verifiable history. 
Non-pivotal candidate configurations are called \emph{tentative}. Intuitively, the next lemma
states that in the rest of the proofs we can almost always consider only pivotal configurations.

\begin{lemma}
\label{lm:tentativeconfigurations} 
Tentative configurations are never used, i.e.: 
\begin{itemize}
    \item No correct process will ever make a request to a tentative configuration;
    \item Tentative configurations cannot be installed;
    \item A correct process will never invoke \text{FSVerify} with timestamp $\textit{height}(C)$ for any tentative configuration $C$;
    \item No correct process will ever make a request to a tentative configuration;
    \item A correct process will never broadcast any message via the weak reliable broadcast primitive in a tentative configuration.
\end{itemize}
\end{lemma}
\begin{proof}
Follows directly from the algorithm. Processes only operate on configurations that we obtained by invoking the function $\text{HighestConf}(h)$.
\end{proof}

\begin{lemma}
\label{lm:pivotalconfigurations}
If $C \sqsubseteq \text{HighestConf}(h)$, where $C$ is a pivotal configuration and $h$ is the local history of a correct process, then $C \in h$.
\end{lemma}
\begin{proof}
Follows from the definition of a pivotal configuration and the requirement that all verifiable histories are related by containment.
\end{proof}

\begin{theorem}
\label{thm:dynamicvalidity}
Only candidate configurations can be installed.
\end{theorem}
\begin{proof}
Follows directly from the algorithm. A correct process will not install a configuration until it is in its local history.
\end{proof}

In the algorithm it is possible for a configuration to be installed after it became superseded. 

We call a configuration  \emph{properly installed} if it was installed
while being active (i.e., not superseded). We use this definition in the proofs of the next few lemmas.

\begin{lemma}
\label{lm:properlyinstalled}
The lowest properly installed configuration higher than some configuration $C$ is the
first (in global time order) installed configuration higher than $C$.
\end{lemma}
\begin{proof}
Let $N$ be the lowest properly installed configuration higher than $C$. If some configuration
higher than $N$ was installed earlier, then $N$ would not be properly installed. If some
configuration between $C$ and $N$ were installed earlier, then N would not be the lowest.
\end{proof}

\begin{lemma}
\label{lm:keyupdate}
If a pivotal configuration $C$ is superseded, then there is no quorum
of processes in that configuration capable of signing messages with timestamp $\textit{height}(C)$. More formally: $\nexists Q \in \textit{quorums}(C): \forall q \in Q: \textit{st}_q \leq \textit{height}(C)$.
\end{lemma}
\begin{proof}
Let $N$ be the lowest properly installed configuration higher than $C$. Let us consider
the moment when $N$ was installed. By the algorithm, all correct processes in some quorum
$Q_N \in \textit{quorums}(N)$ had to broadcast \textbf{UpdateComplete} messages before $N$ was installed. Since $N$ was not yet superseded at that moment, there was at least
one correct process $q_N \in Q_N$.
By Lemma~\ref{lm:pivotalconfigurations}, C was in $r_N$’s local history whenever it performed state transfer to any
configuration higher than $C$. By the protocol, a correct process only advances its $C_{cur}$
variable after executing the state transfer protocol or right before
installing a configuration. Since no configurations between $C$ and $N$
were yet installed, $q_N$ had to pass through C in its state transfer protocol and to receive
\textbf{UpdateReadResp} messages from some quorum $Q_C \in quorums(C)$.
Recall that correct processes update their private keys whenever they learn about a
higher configuration and that they will only reply to message
$\langle \textbf{UpdateRead}, sn, C\rangle$ once $C$ is not the highest configuration in their local histories. This means that all correct processes in $Q_C$ actually had to update their
private keys before $N$ was installed, and, hence, before C was superseded. By the quorum
intersection property, this means that in each quorum in C at least one process updated its
private key to a timestamp higher than $\textit{height}(C)$ and will not be capable of signing messages
with timestamp $\textit{height}(C)$ even if it becomes Byzantine.
\end{proof}

Let us now proceed with the correctness proof of \textit{ConfigLA} and \textit{HistLA}.
Their implementations, in fact, only differ in parameters, what means the proof can be generalized.

In the next part of the proof we consider the following notation for verifiable outputs: \\if $\sigma = \langle vs, \textit{proposenAcks}, \textit{confirmAcks}, h, \sigma_{h} \rangle$ is a certificate for output of \textit{ConfigLA} or \textit{HistLA}, we write $\sigma.\textit{vs}$, $\sigma.\textit{proposeAcks}$, $\sigma.\textit{confirmAcks}$, $\sigma.h$ and $\sigma.\sigma_{h}$ to denote the access to the corresponding parts of certificate $\sigma$.

\begin{lemma}
\label{lm:statetransfercorrectness}
If $\sigma$ is a valid certificate for value $v$ returned from an ABLA object (\textit{ConfigLA} or \textit{HistLA}), then for each active installed configuration $D$ such that $\text{HighestConf}(\sigma.h) \sqsubset D$, there is a quorum $Q_D \in \textit{quorums}(D)$ such that for each correct process in $Q_D$: $\sigma.vs.\text{firsts}()~\subseteq~values.\text{firsts}()$.
\end{lemma}
\begin{proof}
Let $C = \text{HighestConf}(\sigma.h)$. We proceed by induction on the sequence of all properly
installed configurations higher than $C$. Let us denote this sequence by  $\widehat{C}$. Other configurations are not interesting simply because there is no such moment in time when they are simultaneously active and installed.

Let $N$ be lowest configuration in $\widehat{C}$. Let $Q_C \in \textit{quorums}(C)$ be a quorum of processes whose signatures are in $\sigma.\textit{proposeAcks}$. Since $N$ is installed, there is a quorum $Q_N \in \textit{quorums}(N)$ in which all correct processes broadcast $\langle \textbf{UpdateComplete}, N\rangle$. For each correct process
$q_N \in Q_N$, $q_N$ passed with its state transfer protocol through configuration $C$ and received
\textbf{UpdateReadResp} messages from some quorum of processes in $C$. Note that at that moment
$N$ wasn’t yet installed, and, hence $C$ wasn’t yet superseded. By the quorum intersection
property, there is at least one correct process $q_C \in Q_C$ that sent an \textbf{UpdateReadResp}
message to $q_N$. Because $q_C$ will only react to $q_N$’s message after
updating its private keys, it had to sign $\langle \textbf{ProposeResp}, \sigma.vs \rangle$ before sending reply to $q_N$, which means that the \textbf{UpdateReadResp}
message from $q_C$ to $q_N$ must have contained a set of values that includes all values from $\sigma.vs$.
This proves the base case of the induction.

Let us consider any configuration $D \in \widehat{C}$ such that $N \sqsubset D$. Let $M$ be the highest
configuration in $\widehat{C}$ such that $N \sqsubseteq M \sqsubset D$ (in other words, the closest to $D$ in $\widehat{C}$ ). Assume that
the statement holds for $M$, i.e., while $M$ was active, there was a quorum $Q_M \in \textit{quorums}(M)$
such that for each correct process in $Q_M$: $\sigma.vs.\text{firsts}() \subseteq \textit{values}.\text{firsts}()$. Similarly to the base case, let
us consider a quorum $Q_D \in quorums(D)$ such that every correct process in $Q_D$ reliably
broadcast $\langle \textbf{UpdateComplete},D \rangle$ before $D$ was installed. For each correct process $q_D \in Q_D$,
by the quorum intersection property, there is at least one correct process in $Q_M$ that sent an
\textbf{UpdateReadResp} message to $r_D$. This process attached its $values$ to the message, which
contained values from $\sigma.vs$. This proves the induction step and completes the proof.
\end{proof}

\begin{lemma}
\label{lm:oneconfigcomparability}
If $\sigma_1$ is a valid certificate for output value $v_1$ of an ABLA object (\textit{ConfigLA} or \textit{HistLA}),
and $\sigma_2$ is a valid certificate for output value $v_2$ of the same ABLA object,
and $\text{HighestConf}(\sigma_1.h)~=~\text{HighestConf}(\sigma_2.h)$, then $v_1$ and $v_2$ are comparable.
\end{lemma}
\begin{proof}
Let $C = \text{HighestConf}(\sigma_1.h) = \text{HighestConf}(\sigma_2.h)$. 
By definition, $\sigma$ is a valid certificate for $v$
implies that $\text{VerifyOutputValue}(v, \sigma) = \textit{true}$. By the implementation,
$\sigma.h_1$ and $\sigma.h_2$ are verifiable histories. Therefore, $C$ is a pivotal configuration.

The set $\sigma_1.\textit{confirmAcks}$ contains signatures from a quorum of replicas of configuration $C$,
with timestamp $\textit{height}(C)$. Each of these signatures had to be produced after each of the
signatures in $\sigma_1.\textit{proposeAcks}$ because they sign the message $\langle \textbf{ConfirmResp}, \sigma_1.proposeAcks\rangle$. Combining this with the statement of Lemma~\ref{lm:keyupdate}, it
follows that at the moment when the last signature in the set $\sigma_1.\textit{proposeAcks}$ was created, the
configuration $C$ was active (otherwise it would be impossible to gather $\sigma_1.\textit{confirmAcks}$). We
can apply the same argument to the sets $\sigma_2.\textit{proposeAcks}$ and $\sigma_2.\textit{confirmAcks}$.

It follows that there are quorums $Q_1, Q_2 \in \textit{quorums}(C)$ and a moment in time $t$
such that: (1) $C$ is not superseded at time $t$, (2) all correct replicas in $Q_1$ signed
message $\langle \textbf{ProposeResp}, \sigma_1.vs\rangle$ before $t$, and (3) all correct replica in $Q_2$ signed message
$\langle \textbf{ProposeResp}, \sigma_2.vs\rangle$ before $t$. Since $C$ is not superseded at time $t$, there must be a correct
replica in $Q_1 \cap Q_2$ (due to quorum intersection), which signed both $\langle \textbf{ProposeResp},  \sigma_1.vs \rangle$ and
$\langle \textbf{ProposeResp},  \sigma_2.vs\rangle$. Since correct replicas only sign \textbf{ProposeResp}
messages with comparable sets of values, $ \sigma_1.vs.\text{firsts}()$ and $ \sigma_2.vs.\text{firsts}()$ are related by containment, i.e., either $ \sigma_1.vs.\text{firsts}() \subseteq  \sigma_2.vs.\text{firsts}()$
or $\sigma_2.vs.\text{firsts}() \subset  \sigma_1.vs.\text{firsts}()$. Hence, $v_1 = \bigsqcup \sigma_1.vs.\text{firsts}()$ and $v_2 = \bigsqcup \sigma_2.vs.\text{firsts}()$ are comparable.
\end{proof}

\begin{theorem}
\label{thm:comparability}
All verifiable output values produced by the same ABLA object (\textit{ConfigLA} or \textit{HistLA}) are comparable.
\end{theorem}
\begin{proof}
Let $\sigma_1$ be a valid certificate for output value $v_1$, and $\sigma_2$ be a valid certificate for output value $v_2$.
Also, let $C_1 = \text{HighestConf}(\sigma_1.h)$ and $C_2 = \text{HighestConf}(\sigma_2.h)$. Since $\sigma_1.h$ and $\sigma_2.h$ are verifiable histories, both $C_1$ and $C_2$ are pivotal by definition.

If $C_1 = C_2$, $v_1$ and $v_2$ are comparable by Lemma~\ref{lm:oneconfigcomparability}.

Consider the case when $C_1 \not = C_2$. Without loss of generality, assume that $C_1 \sqsubset C_2$.
Let $Q_1 \in quorums(C_2)$ be a quorum of replicas whose signatures are in $\sigma_2.\textit{proposeAcks}$. 
Let $t$ be the moment when first correct replica signed $\langle \textbf{ProposeResp}, \sigma_2.vs \rangle$. 
Correct replicas only start processing user requests in a configuration when this configuration is installed. Therefore, by Lemma~\ref{lm:statetransfercorrectness}, at time $t$ there was a quorum of replicas
$Q_2 \in quorums(C_2)$ such that for every correct replica in $Q_2$: $\sigma_1.vs.\text{firsts}() \subseteq \textit{values}.\text{firsts}()$. 
By the quorum intersection property, there must be at least one correct replica in $Q_1 \cap Q_2$. Hence, $\sigma_1.vs.\text{firsts}() \subseteq \sigma_2.vs.\text{firsts}()$
and $v_1 = \bigsqcup \sigma_1.vs.\text{firsts}() \sqsubseteq v_2 = \bigsqcup \sigma_2.vs.\text{firsts}()$.
\end{proof}

\begin{theorem}
\label{thm:blasafety}
Both \textit{ConfigLA} and \textit{HistLA} objects satisfy properties of Adjustable Byzantine Lattice Agreement: \emph{ABLA-Validity}, \emph{ABLA-Verifiability}, \emph{ABLA-Inclusion}, \emph{ABLA-Comparability}. 
\end{theorem}
\begin{proof}
Let us show that all properties hold:
\begin{itemize}
    \item ABLA-Validity follows directly from the implementation;
    \item ABLA-Verifiability also follows directly from the implementation; 
    \item ABLA-Inclusion follows from the fact that a correct process always includes proposed value $v$ in its local set $\textit{values}$ and this set only grows;
    \item ABLA-Comparability follows from the Theorem~\ref{thm:comparability}.
\end{itemize}
\end{proof}

\begin{theorem}
\label{thm:histories}
All verifiable histories returned from $\textit{HistLA}.\text{Propose}\atremove{(\dots)}$ (line~\ref{line:posmain:historypropose}) are related by containment and all configurations in any such history are comparable.
\end{theorem}
\begin{proof}
It follows from the fact that both \textit{ConfigLA} and \textit{HistLA} meet specification of ABLA and the only valid input values for \textit{HistLA} are sets that contain only the output values of \textit{ConfigLA} that are comparable.
\end{proof}

Now, in a similar way, we prove that implementation of \textit{TxVal} object satisfies the specification of Transaction Validation.

Similiarly, to verifiable outputs  we consider the following notation for verifiable transaction sets: 
\\
if $\sigma = \langle \textit{sentTxs}, \textit{validationAcks}, \textit{confirmAcks}, h, \sigma_{h} \rangle$ is a certificate for transaction set produced by \textit{TxVal}, we write $\sigma.\textit{sentTxs}$, $\sigma.\textit{validationAcks}$, $\sigma. \textit{confirmAcks}$, $\sigma.h$ and $\sigma.\sigma_{h}$ to denote the access to the corresponding parts of certificate~$\sigma$.

For the sake of simplicity, when we do not care about transaction certificate sometimes,  we write $tx \in \textit{sentTxs}$ (or $\textit{conflictTxs}$ resp.) instead of $\langle tx, \sigma_{tx} \rangle \in \textit{sentTxs}$ ($\textit{conflictTxs}$ resp.).

\begin{lemma}
\label{lm:noconflictinvalue}
If $\sigma$ is a valid certificate for output transaction set $\textit{txs}$ of the  \textit{TxVal} object, then no two transactions in $txs$ conflict.
\end{lemma}
\begin{proof}
Let us assume that process obtained a valid certificate for a transaction set $\textit{txs}$ where two transactions $tx_1 \in \textit{txs}$ and $tx_2 \in \textit{txs}$ conflict. 
We denote $C = \text{HighestConf}(\sigma.h)$. As certificate is valid, then $\langle tx_1, \sigma_{tx_1}\rangle \in \sigma.\textit{sentTxs}$ and $\langle tx_2, \sigma_{tx_2} \rangle \in \sigma.\textit{sentTxs}$ and corresponding certificates are valid.
At the same time, we have that $\exists Q \in \textit{quorums}(C)$ such that every process $q \in Q$ signed a message $\langle \textbf{ValidateResp}, \sigma.\textit{sentTxs}, \textit{conflictTxs}_{q} \rangle$.
For each such process $q$ and message it signed, there is a corresponding information (process identifier $q$, signature, set $\textit{conflictTxs}_{q}$) stored in $\sigma.\textit{validationAcks}$. 
As $tx_1$ and $tx_2$ appear in output transaction set $txs$, then there is no $q$ such that $tx_1 \in \textit{conflictTxs}_{q}$ or $tx_2 \in \textit{conflictTxs}_{q}$. 
However, at the moment when the last signature in the set $\sigma.\textit{validationAcks}$ was created, the configuration $C$ was active (otherwise it would be impossible to gather $\sigma.\textit{confirmAcks}$). 
As $C$ was active, there was a correct process $p \in Q$ (due to the quorum intersection property) that should have detect the conflict between any two transactions in $\sigma.\textit{sentTxs}$, and place them into $\textit{conflictTxs}_p$ along with their certificates.
This contradicts with our initial assumption. 
Thus, no two transactions in $\textit{txs}$ conflict.
\end{proof}

\begin{lemma}
\label{lm:noconflictoneconfig}
If $\sigma_1$ is a valid certificate for output transaction set $\textit{txs}_1$ of the \textit{TxVal} object,
and $\sigma_2$ is a valid certificate for output transaction set $\textit{txs}_2$ of the \textit{TxVal} object,
and $\text{HighestConf}(\sigma_1.h)=\text{HighestConf}(\sigma_2.h)$, then $\nexists tx_1 \in \textit{txs}_1$ such that $\textit{tx}_1$ conflicts with some $tx_2 \in \textit{txs}_2$.
\end{lemma}
\begin{proof}
Let us denote $\text{HighestConf}(\sigma_1.h)$ as $C$.
As $\sigma_1$ is a valid certificate for $\textit{txs}_1$, then $C$ was active when the last acknowledgment in $\sigma_1.\textit{validationAcks}$ was gathered.
The same is true for $\sigma_2$ and $\textit{txs}_2$.

For output transaction set $\textit{txs}_1$ validations in $\sigma_1.\textit{validationAcks}$ were collected from some quorum $Q_1 \in \textit{quorums}(C)$, and for $\textit{txs}_2$  from  $Q_2 \in \textit{quorums}(C)$. 
As $C$ was active at the moments of time when last signatures in $\sigma_1.\textit{validationAcks}$ and in $\sigma_2.\textit{validationAcks}$ were created, then there existed a correct process $q \in Q_1 \cap Q_2$ due to quorum intersection property.
Without loss of generality, we assume that $q$ firstly received and validated message $\langle \textbf{ValidateReq}, \sigma_1.\textit{sentTxs} \rangle$. 
As $q$ was a correct process, it added all the transactions from  $\textit{sentTxs}_1$ to its local estimate $\textit{seenTxs}_q$. 
Then, when $q$ received a message  $\langle \textbf{ValidateReq}, \sigma_2.\textit{sentTxs} \rangle$ it should have revealed all the
transactions in $\sigma_1.\textit{sentTxs}$ that conflict with the ones it had seen before
(including the ones from $\sigma_1.\textit{sentTxs}$), and it placed conflicting transactions along with their certificates 
into the set $\textit{conflictTxs}$ according to the algorithm. Hence, no transaction $tx_1 \in \textit{txs}_1$ 
conflicts with any transaction $tx_2 \in \textit{txs}_2$.
\end{proof}

\begin{lemma}
\label{lm:statetransfercorrectnesstansfervalidation}
Let $\textit{txs}$ be a verifiable transaction set and $\sigma$ be a matching certificate. Then, for any installed and active configuration $C$ such that $\text{HighestConf}(\sigma.h) \sqsubset C$ the following condition is met: $\exists Q: Q \in 
\textit{quorums}(C)$ and $\forall q \in Q$ and $q$ is correct: $\sigma.\textit{sentTxs} \subseteq \textit{seenTxs}_q$.
\end{lemma}
\begin{proof}
The proof of this fact is constructed similarly to the proof of the Lemma~\ref{lm:statetransfercorrectness}.
\end{proof}

\begin{lemma}
\label{lm:noconflictdifconfigs}
Let $\textit{txs}_1$ be a verifiable transaction set of a Transaction Validation object \textit{TxVal} with a corresponding certificate $\sigma_1$ and $\textit{txs}_2$ be a verifiable transaction set of the same Transaction Validation object \textit{TxVal} with a corresponding certificate $\sigma_2$. Then, $\nexists tx_1 \in \textit{txs}_1$ such that $tx_1$ conflicts with some $tx_2 \in \textit{txs}_2$.
\end{lemma}
\begin{proof}
Let $C_1 = \text{HighestConf}(\sigma_1.h)$ and $C_2 = \text{HighestConf}(\sigma_2.h)$. 
As $\sigma_1$ is a valid certificate for $\textit{txs}_1$, then $C_1$ was active at the moment of time, when the last signature was gathered for $\sigma_1.\textit{validationAcks}$ (otherwise, it would be impossible to collect $\sigma_1.\textit{confirmAcks}$). The same is true for $\sigma_2$, $\textit{txs}_2$ and $C_2$ respectively.

If $C_1 = C_2$, then by Lemma \ref{lm:noconflictoneconfig} transactions from $\textit{txs}_1$ don't conflict with transactions from $\textit{txs}_2$.

Now, we consider the case when $C_1 \not= C_2$. Similarly to the Theorem~\ref{thm:comparability} we can assume that $C_1 \sqsubset C_2$ (w.l.o.g.). 
Let $Q_1 \in \textit{quorums}(C_2)$ be a quorum of processes that validated $\textit{sentTxs}_2$ and whose signatures are in $\sigma_2.\textit{validationAcks}$. 
Let $t$ be the moment of time when first correct process, whose signature is in $\sigma_2.\textit{validationAcks}$, validated $\textit{sentTxs}_2$. 
Correct processes only start validation when the configuration is installed. 
Therefore, by Lemma~\ref{lm:statetransfercorrectnesstansfervalidation}, at time $t$ there was a quorum of processes $Q_2 \in \textit{quorums}(C_2)$ such that for every correct process in $Q_2$: $\textit{sentTxs}_1 \subseteq \textit{seenTxs}$.
As $\textit{quorums}(C)$ is a dissemination quorum system at time $t$, there must be at least one correct process $q \in Q_1 \cap Q_2$. As $q$ is correct and $\textit{sentTxs}_1 \subseteq \textit{seenTxs}$, then $q$ marked all the transactions from $\textit{sentTxs}_2$ that conflict with $\textit{sentTxs}_1$ in accordance with algorithm.
Taking this into account and the fact that given certificates $\sigma_1$ and $\sigma_2$ are valid for $\textit{txs}_1$ and $\textit{txs}_2$ respectively, we obtain, that no transaction $tx_1 \in \textit{txs}_1$ conflicts with any transaction $tx_2 \in \textit{txs}_2$.
\end{proof}

\begin{theorem}
\label{thm:TVsafety}
Our implementation of \textit{TxVal} objects satisfies the properties of Transaction Validation: \emph{TV-Verifiability}, \emph{TV-Validity}, \emph{TV-Inclusion}. 
\end{theorem}
\begin{proof}
TV-Verifiability property follows from the implementation. 

TV-Inclusion follows from the fact that a correct process $p$ always adds transaction $tx$ to its local set $\textit{seenTxs}_p$ this set only grows and doesn't contain conflicting transactions issued by $p$ while $p$ is correct.

TV-Validity property follows from Lemma~\ref{lm:noconflictinvalue} and Lemma~\ref{lm:noconflictdifconfigs}.

\end{proof}

\begin{theorem}
\label{thm:noconflict}
There are no two conflicting transactions in any configuration in any verifiable history.
\end{theorem}
\begin{proof}
All verifiable histories are returned from \textit{HistLA}. The only configurations that are contained in any verifiable history are the ones returned from \textit{ConfigLA}.
At the same time, the only possible way for a transaction to be included in a configuration is to be an element of some verifiable transaction set returned by \textit{TxVal}. 
According to the Theorem~\ref{thm:TVsafety}, the only transactions that are returned from \textit{TxVal} are non-conflicting. Hence, no two conflicting transactions can become elements of a verifiable configuration (verifiable output value of \textit{ConfigLA}), and, consequently, any verifiable history contains only configurations with non-conflicting transactions.
\end{proof}

\begin{theorem}
\label{thm:cryptosafety}
\Name satisfies the \emph{Consistency} property of an asset transfer system. 
\end{theorem}
\begin{proof}
If a correct process access its local copy of the state $T_p$, then it reads the last configuration it installed, that, according to implementation, is an element of a local verifiable history. 
From Theorem~\ref{thm:noconflict} we know that any configuration from any verifiable history doesn't contain conflicting transactions.
According to the implementation the only configurations that are installed are the candidate configurations.
Thus, when a correct process accesses its $T_p$, it contains only non conflicting transactions.

The fact that all transactions that are contained in $T_p$ of a correct process are verifiable follows from the implementation of the algorithm.

The fact that for any two correct processes $p$ and $q$ at times $t$ and $t'$ either $T_p(t) \subseteq T_q(t')$ or $T_q(t') \subseteq T_p(t)$ follows from the fact that all installed configurations are related by containment.
\end{proof}

\myparagraph{Monotonicity.} The monotonicty property says that as long as process $p$ is correct its local estimate of confirmed transactions $T_p$ can only grow with time.
The proof of this property is simple and short.

\begin{theorem}
\label{thm:cryptoconsistency}
\Name satisfies the \emph{Monotonicity} property of an asset-transfer system.
\end{theorem}
\begin{proof}
This statement follows from the fact that $T_p$ is the last configuration installed by $p$ (as long as $p$ is correct) 
and the correct process only installs greater configurations (w.r.t. $\subseteq$) over time.
Thus, as long as $p$ is correct, $T_p$ can only grow with time.
\end{proof}

\myparagraph{Validity.} In this subsection we focus on the validity property. Our goal is to prove that if an operation $\code{Transfer}(tx, \sigma_{tx})$ is invoked by a forever-correct process $p$, then eventually $p$ includes transaction $tx$ in its $T_p$.

Note, that even Byzantine processes may obtain verifiable objects (e.g. verifiable transaction set) from the protocols that may differ from the presented implementations.
However, it should be noted, that due to the implementation of corresponding verifying functions the guarantees for such objects are similar to the ones that are returned by correct processes from implemented protocols (except, maybe TV-Inclusion and ABLA-Stability).
Taking this into account, for simplicity we assume that if a Byzantine process obtained some verifiable object it had done it via the corresponding protocol implemented in this work.

\begin{lemma}
\label{lm:TVev}
If a forever-correct process $p$ invokes $\textit{TxVal}.\text{Validate}(tx, \sigma_{tx})$ (at line~\ref{line:posmain:validate}), then for some configuration $C'$ every verifiable set of transactions $\textit{txs}$ returned from $\textit{TxVal}.\text{Validate}\atremove{(\dots)}$ with corresponding certificate $\sigma$, such that $C' \sqsubset \text{HighestConf}(\sigma.h)$ contains $tx$, i.e. $tx \in \textit{txs}$.
\end{lemma}
\begin{proof}
In case a forever-correct process $p$ returns a verifiable transaction set $txs$ with a corresponding certificate $\sigma$ from $\textit{TxVal}.\text{Validate}(tx, \sigma_{tx})$, the proof of this fact follows from Lemma~\ref{lm:statetransfercorrectnesstansfervalidation}.

If $p$ is ``stuck'' and it does not return from $\textit{TxVal}.\text{Validate}(tx, \sigma_{tx})$, since $p$ is a forever-correct process it constantly tries to reach every ``greatest'' configuration it learns about by sending a message $\langle \textbf{ValidateReq}, \textit{sentTxs}, \textit{sn} \rangle$, where $tx\in \textit{sentTxs}$.
As the number of the unique members in the system is limited, starting from some configuration $C'$ every correct member will include $tx$ in its local estimate $\textit{correctTxs}$.
From this moment every verifiable transaction set $\textit{txs}$, obtained by any process must include $tx$.
\end{proof}

Intuitively, the next lemma states that the number of unique verifiable transaction sets produced in some configuration $C$ is finite.
\begin{lemma}
\label{lm:vtsfiniteinconfig}
Let $C$ be some configuration and let $T$ be a set defined as follows:
$T = \{ \langle \textit{txs}, \sigma_{txs} \rangle\ | $ $\textit{txs}$ is a verifiable transaction set, $\sigma_{txs}$ is a matching certificate, such that $\text{HighestConf}(\sigma_{txs}.h) = C \}$.
Let $T'$ be a set defined as $T' = \{ \textit{txs}~|~\langle \textit{txs}, \sigma_{txs} \rangle \in T \}$. Then, $T'$ is a finite set.
\end{lemma}
\begin{proof}
Let us suppose that $T'$ is an infinite set. 
Considering this and the fact that any $\textit{txs} \in T'$ is verifiable, we obtain that $C$ was active and installed from some moment of time. 
However, the number of well-formed and non-conflicting transactions in the system  while $C$ is active and installed is finite, due to the fact that number of participants is finite and total system stake is also finite.
From the fact that set of such transactions is finite, we obtain that set $T'$ is finite as well. Contradiction.
\end{proof}

\begin{lemma}
\label{lm:vtswithouttx}
If a forever-correct process invokes $\code{Transfer}(\textit{tx}, \sigma_{tx})$, then number of verifiable transaction sets $\textit{txs}$, such that $tx \not\in \textit{txs}$ is finite. 
\end{lemma}
\begin{proof}
The proof of this fact follows from the Lemma~\ref{lm:TVev} and the Lemma~\ref{lm:vtsfiniteinconfig}.
As a forever-correct process invoked  $\code{Transfer}(\textit{tx}, \sigma_{txs})$ operation, it later invoked $\textit{TxVal}.\text{Validate}(tx, \sigma_{tx})$ at line~\ref{line:posmain:validate} . 
By Lemma~\ref{lm:TVev}, for some configuration $C'$ any verifiable transaction set $\textit{txs}$ with corresponding certificate $\sigma.h$ such that $C' \sqsubset \textit{HighestConf}(\sigma.h)$ contains $tx$. By Lemma~\ref{lm:vtsfiniteinconfig}, number of verifiable transaction sets obtained in any configuration $C: C \sqsubseteq C'$ is finite. 
Also, the number of such configurations $C$ is finite. 
Hence, the number of all verifiable transaction sets $\textit{txs}$ such that $tx \not\in \textit{txs}$ is finite.
\end{proof}

\begin{lemma}
\label{lm:configwithouttx}
If a forever-correct process invokes $\code{Transfer}(\textit{tx}, \sigma_{tx})$, then number of verifiable configurations $C$, such that $tx \not\in C$ is finite.
\end{lemma}
\begin{proof}
If a verifiable configuration $C$ doesn't contain $tx$, then it was obtained as a merge of verifiable transaction sets that don't contain $tx$.
By Lemma~\ref{lm:vtswithouttx} the number of the verifiable transaction sets $txs$ that don't contain $tx$ is finite.
Consequently, the total number of verifiable configurations $C$, such that $tx \not\in C$ is finite.  
\end{proof}

\begin{lemma}
\label{lm:historywithouttx}
If a forever-correct process invokes $\code{Transfer}(\textit{tx}, \sigma_{tx})$, then number of verifiable histories $H$, such that $\not\exists C \in H: tx \in C$ is finite.
\end{lemma}
\begin{proof}
The proof of this fact is similar to the proof of Lemma~\ref{lm:configwithouttx}.
If a verifiable history $H$ doesn't contain configuration $C$ such that $tx \in C$, then it was obtained as a merge of verifiable configurations that don't contain $tx$.
By Lemma~\ref{lm:configwithouttx} the number of the verifiable configurations $C$ that don't contain $tx$ is finite.
As a consequence,  the total number of verifiable histories $H$, such that $\not\exists C \in H: tx \in C$ is finite.
\end{proof}

\begin{lemma}
\label{lm:forevercorrectconfig}
If a configuration $C$ is installed, then some forever-correct process will eventually install a configuration $C'$ such that $C \sqsubseteq C'$.
\end{lemma}
\begin{proof}
We prove this fact by contradiction. 
Let us consider a potentially infinite set of all installed configurations in the system.
All installed configurations are related by containment, so they form a sequence. 
We can denote such a sequence as $\widehat{C} = \{C_{init}, \dots C_{i}, C_{i+1},\dots\}$. We assume that no forever-correct process installs a configuration greater or equal to $C_{n} \in \widehat{C}$.
It means, that only correct (and maybe Byzantine) members took part in the installation process of a configuration $C_{n}$ and greater ones.
We denote a sequence of installed configurations greater or equal to $C_n$ as $\widehat{C}_n$.
It's easy to see, that sequence $\widehat{C}_n$ is a suffix of the sequence $\widehat{C}$.
For every $C_i \in \widehat{C}_n$ we introduce $S_i$ -- a set of correct processes that installed $C_i$ (i.e., triggered event $\text{InstalledConfig}(C_i)$).

If $\widehat{C}_n$ is finite, then there exists a configuration $C_{max}$ that is installed but never superseded.
As no configuration greater than $C_{max}$ is installed, due to the property of weak uniform reliable broadcast, every forever-correct process delivers the same set of messages that were delivered any correct process in $C_{max}$. Then, any forever-correct process eventually delivers quorum of messages (in $C_{max}$) $\langle \textbf{UpdateComplete}, C_{max} \rangle$.
As $C_{max}$ is installed, then some correct process triggered an event $\text{InstalledConfig}(C)$. Before that it should have delivered a quorum of $\langle \textbf{UpdateComplete}, C_{max} \rangle$ messages. 
As $C_{max}$ is not superseded at that time and due to configuration availability and quorum intersection property, there must have been a process $q$ that broadcast $\langle \textbf{UpdateComplete}, C_{max} \rangle$ and is correct in $C_{max}$. It implies that $q$ previously delivered verifiable history $h$ such that $C_{max} \in h$.
Then, $q$ delivered all the messages that are required for installation of configuration $C_{max}$.
As $C_{max}$ is never superseded $q$ is forever-correct.
Hence, $C_{max}$ is installed by a forever-correct process $q$ and $C_n \sqsubseteq C_{max}$. It contradicts with our assumption.

Now, consider the case when $\widehat{C}_n$ is infinite.
We suppose that no forever-correct process installed a configuration from $\widehat{C}_n$.
It means, that no forever-correct process delivers messages required for installation of $C_i$ in $\widehat{C}_n$: message $\langle \textbf{NewHistory}, h, \sigma_h\rangle$, such that $\sigma_h$ is a valid certificate for $h$ and $C_i \in h$, and/or quorum of messages (in $C_i$) $\langle \textbf{UpdateComplete}, C_i\rangle$.  
It implies that starting from some installed configuration $C_j$, such that $C_i \sqsubset C_j$ every process from $S_i$ turns Byzantine.
Such reasoning is used for any $C_i \in \widehat{C}_n$.
However, the number of potential members of the protocol is finite, so the number of the processes that can turn Byzantine is finite. So, after some configuration $C_k$ is installed the only correct processes left in the system are forever-correct processes. As $C_k$ is installed, then some forever-correct process should have installed it.

We assumed that no forever-correct process installs a configuration greater or equal to installed configuration $C$ and came to a contradiction.
This proves the lemma.
\end{proof}

\begin{lemma}
\label{lm:forevercorrectconfigall}
If a configuration $C$ is installed, then any forever-correct process will eventually install a configuration $C'$ such that $C \sqsubseteq C'$.
\end{lemma}
\begin{proof}
Follows directly from the Lemma~\ref{lm:forevercorrectconfig} and the guarantees that are provided by weak reliable broadcast. 
From the Lemma~\ref{lm:forevercorrectconfig} some forever-correct process $p$ eventually installs a configuration $C'$, such that $C \sqsubseteq C'$. 
As $p$ is a forever-correct process, then it previously delivered  a message $\langle \textbf{NewHistory}, H, \sigma_H \rangle$ via weak reliable broadcast, where $\sigma_H$ is a valid certificate for $H$, $C' = \text{HighestConf}(H)$, and also delivered a set of messages $\langle \textbf{UpdateComplete}, C' \rangle$ from some $Q \in \textit{quorums}(C')$ also via weak reliable broadcast primitive. 
As these messages were delivered via WRB by a forever-correct process, then any other forever-correct process will eventually deliver this set of messages and install configuration $C'$ or a greater one.
\end{proof}

\begin{lemma}
\label{lm:statetransfercompletion}
Every state transfer protocol that is executed by a forever-correct process terminates.
\end{lemma}
\begin{proof}
Let us consider a forever-correct process $p$ that executes state transfer protocol. Any verifiable history $h$ on which correct processes execute state transfer protocol is finite, therefore We should show that $\not\exists C \in h$ such that $p$ waits for replies from some of $\textit{quorums}(C)$ infinitely long (line~\ref{line:statetrans:wait}). 
In this lemma we proceed by contradiction. Let us assume, that there is a configuration
$C$, such that $C \sqsubset C_n$ and $C_{cur} \sqsubseteq C$, for which $p$ waits for replies indefinitely. 
If there is no such $ Q \in \textit{quorums}(C)$ that replied correctly, then it means that $C$ is not active and was already superseded by some configuration $C'$.
Hence, sufficiently many processes had already performed state transfer from configuration $C$, hence, $p$ will eventually deliver $\langle \textbf{NewHistory}, H, \sigma_H \rangle$, where $H$ is a verifiable history, and $\sigma_H$ is a matching certificate, such that $C'' \in H$ and $p$ will also deliver $\langle \textbf{UpdateComplete}, C'' \rangle$, where $C' \sqsubseteq C''$ and $C \sqsubset C'$ (follows from Lemma~\ref{lm:forevercorrectconfigall}). Thus, \textbf{wait for} condition (at line~\ref{line:statetrans:wait}) for configuration $C$ will be satisfied. 
\end{proof}

\begin{lemma}
\label{lm:transferliveness}
If a forever-correct process $p$ invokes $\code{Transfer}(tx, \sigma_{tx})$, then eventually configuration $C$ is installed, such that  $tx \in C$.
\end{lemma}
\begin{proof}
Let us suppose, for contradiction, that no configuration $C$, such that $tx \in C$ is installed. 
Due to the Lemma~\ref{lm:configwithouttx} number of verifiable configurations $C$ such that $tx \not\in C$ is finite. Also, by the Lemma~\ref{lm:historywithouttx} number of verifiable histories $H$ such that $\not\exists C \in H$ such that $tx \in C$ is also finite.
The only configurations that can be installed in the system are verifiable configurations that are part of a verifiable history. 
Combining the results of these Lemmas with an assumption that no configuration $C$, such that $tx \in C$ is installed, we obtain that starting from some moment of time $t$ no greater configuration can be installed. 
Thus, there exists greatest configuration $C_m$ that is installed in the system. 
By Lemma~\ref{lm:forevercorrectconfigall} eventually all forever-correct processes install configuration $C: C_m \sqsubseteq C$. As $C_m$ is the greatest configuration, then $C = C_m$.
As $C_m$ can't be superseded, then the number of well-formed and non-conflicting transactions is finite and correct processes will eventually have equal sets $\textit{seenTxs}$. 
Similarly, the local sets $\textit{values}$ of correct processes of \textit{ConfigLA} and \textit{HistLA} eventually converge in terms of inputs (certificates might be different).
Hence, as the mentioned sets of all correct participants eventually converge and no greater configuration than $C_m$ is installed, a forever-correct process $p$ should be able to return a verifiable history $H$ such that $tx \in C$ and $C \in H$. 
In particular $tx \in C_{tx}$, where  $C_{tx} = \text{HighestConf}(H)$.
Then, $p$ broadcasts verifiable history $H$ via weak reliable broadcast. 
As no configuration greater than $C_m$ is installed, all correct processes that are left in the system are forever-correct processes. 
Due to the guarantees of weak reliable broadcast, all such processes will eventually deliver verifiable history $H$ with corresponding certificate $\sigma_H$ and trigger state transfer protocol. 
By Lemma~\ref{lm:statetransfercompletion} all these processes will successfully finish this protocol. 
Then, there must exist a quorum $Q \in \textit{quorums}(C_{tx})$ that broadcast $\langle \textbf{UpdateComplete}, C_{tx}\rangle$. 
It means that some correct process will eventually install configuration $C_{tx}$, such that $tx \in C_{tx}$. 
We came to a contradiction, assuming that no configuration $C$, such that $tx \in C$ is installed. It proves the statement of the lemma. 
\end{proof}

\begin{theorem}
\label{thm:pastro_validity}
\Name satisfies the \emph{Validity} property of an asset-transfer system. 
\end{theorem}
\begin{proof}
From Lemma~\ref{lm:transferliveness} it follows, that if a forever-correct process $p$ invokes $\code{Transfer}(tx, \sigma_{tx})$, then eventually a configuration $C_{tx}$ such that $tx \in C$ is installed. 
As $C_{tx}$ is eventually installed, then by Lemma~\ref{lm:forevercorrectconfigall} every forever-correct process, including $p$, will install configuration $C$, such that $C_{tx} \sqsubseteq C$. As $\sqsubseteq$ for configurations is defined as $\subseteq$, $C$ contains $tx$ (i.e. $tx \in C$).
As $p$ is a forever-correct process, $T_p$ is precisely the last configuration installed by $p$, what means that eventually $p$ adds $tx$ to $T_p$.

\end{proof}

\myparagraph{Agreement.} The agreement property tells that if a correct process $p$ at time $t$  includes $tx$ in its local copy of shared state $T_p$, then any forever-correct process $q$ eventually includes $tx$ in $T_q$. 

\begin{theorem}
\label{thm:pastro_agreement}
\Name satisfies the \emph{Agreement} property of an asset-transfer system. 
\end{theorem}
\begin{proof}
If a correct process $p$ at time $t$ added transaction $tx$ to $T_p$, then it installed a configuration $C$, such that $tx \in C$. 
By Lemma~\ref{lm:forevercorrectconfigall}, if configuration $C$ is installed, then any forever-correct process $q$ will eventually install some greater configuration $C'$ (w.r.t. $\subseteq$). 
As $C \sqsubseteq C'$ and $\sqsubseteq = \subseteq$, then $tx \in C'$. It means, that eventually any forever-correct process $q$ adds $tx$ to $T_q$.

\end{proof}

\section{Enhancements and Optimizations}
\label{app:optimizations}
\myparagraph{Communication and storage costs.}
For the sake of simplicity, the current version  of our protocol maintains an ever-growing data set, consisting of transactions, configurations and histories of configurations.
This data is routinely attached to the protocol messages. 
An immediate optimization would be to only attach the \emph{height} of a configuration instead of the configuration itself (e.g., $|C|$ instead of $C$ in message $\langle\textbf{ValidateReq}, \textit{sentTxs}, \textit{sn}, C\rangle$)\atremove{, as our protocol guarantees that all installed configurations are comparable}.
\atadd{Because our protocol guarantees that all installed configurations are comparable, the height can be used to uniquely identify a configuration.}
%
%
Furthermore, to \atreplace{bound}{reduce} the sizes of protocol messages, the processes may only send ``deltas'' containing the updates with respect to the information their recipients should already have.
Similar techniques are often applied in CRDT implementations~\citep{crdt}, where processes also maintain data structures of unbounded size.  
%
%
%


The time complexity of the protocol is linear in the number of concurrently issued valid non-conflicting transactions,  
which in turn is bounded by the total number of processes in the system. 

Similarly to other cryptocurrencies~\cite{ethereum, bitcoin}, in \Name, the replicas are required to store the full history of transactions. The storage complexity is therefore quasi-linear to the number of executed transactions (under the assumptions that certificates are stored efficiently).

A forward-secure \emph{multi-signature} scheme~\citep{drijvers2019pixel} can be used to further reduce the communication and storage costs by representing multiple signatures created by different processes for the same \atreplace{object}{message} as one compact signature.

\myparagraph{Delegation.}
So far our system model did not distinguish between processes acting as \emph{replicas}, 
maintaining copies of the shared data, and \emph{clients}, invoking data operations and interacting with the replicas. 
Many  cryptocurrencies  allow ``light'' clients to only submit transactions without participating in the full protocol. 
In {\Name}, this can be implemented with a \emph{delegation} mechanism:
a participant may disclaim its responsibility for system progress and entrust other participants with its stake.
To give a concrete example of such a mechanism, let us introduce a \emph{delegation transaction} $\textit{dtx} = (p, \delta, \textit{sn})$. 
Here $p \in \Pi$ is the identifier of the process issuing $\textit{dtx}$, 
the map $\delta: \Pi \rightarrow \mathbb{R}$ is the \textit{delegation function}, which specifies the proportion of funds \textit{delegated} to every process $q \in \Pi$ by this transaction, and $\textit{sn}$ is a sequence number. 
%
For every process $p$, there is a special \textit{initial delegation transaction} $\textit{dtx}_{p\_init} = (p, \{(p, 1)\}, 0)$. 
Let $\dD$ be a set of delegation transactions. Delegation transactions $\textit{dtx}_i$ and $\textit{dtx}_j$ conflict iff $\textit{dtx}_i \not= \textit{dtx}_j$, $\textit{dtx}_i.p = \textit{dtx}_j.p$ and $\textit{dtx}_i.\textit{sn} = \textit{dtx}_j.\textit{sn}$. 
Note that Transaction Validation used for normal transactions can be easily adapted for delegation transactions.

%

Consider a composition of $\cC_1\times\cC_2$, where $\cC_1$ is a lattice over sets of transactions $2^{\tT}$ and $\cC_1$ is a lattice over sets of delegation transactions $2^{\dD}$.%
\footnote{Recall that two lattices can be easily composed: given two lattices  $(\cC_1, \sqsubseteq_1)$ and $(\cC_2, \sqsubseteq_2)$, we can define a lattice $(\cC, \sqsubseteq) = (\cC_1 \times \cC_2, \sqsubseteq_1 \times \sqsubseteq_2)$. Given two elements of the composed lattice $x = (c_1, d_1)$ and $y = (c_2, d_2)$, $x \sqsubseteq y$ iff $c_1 \sqsubseteq_1 c_2$ and $d_1 \sqsubseteq_2 d_2$.}
A configuration is now defined as  $C = (T, D)$, where $T \in 2^\tT$ and $D \in 2^\dD$. The stake of a process $q$ in $C$ can be computed as follows:
$\textit{stake}(q, C) = \sum_{tx\in T} tx.\tau(q)-\sum_{tx\in T \wedge tx.p=q} \textit{value}(tx)$,
and its delegated stake as
$\textit{delegatedStake}(q, C) = \sum_{r \in \textit{members}(C)}(\textit{dtx}_r.\delta(q) \cdot \textit{stake}(r, C))$,
where $\textit{dtx}_r = \argmax_{d \in D \land d.p = r} d.sn$.
The configuration-availability assumption (Section~\ref{sec:adversary}) must then be adapted to hold for the \emph{delegated} stake.
Besides, the use of \emph{proportions} in delegation transactions allows us to issue both types of transactions separately and to delegate its stake to  multiple participants.
One can think of the processes with positive delegated-stake as validators.
%
%
For the sake of efficiency, the number of validators should remain relatively small. 
Below we discuss some ways to achieve this.

\myparagraph{Fees and inflation.}
As for now, nothing prevents Byzantine participants from spamming the system with non-conflicting transactions that transfer funds among themselves.
Even though the number of such transactions is bounded, they can still significantly slow down the system progress.
Furthermore, as the owners of processes may lose their key material or stop participating in the protocol and some replicas can get compromised by the adversary, the total amount of funds in the system is bound to shrink over time.
To fight these issues, one may introduce \emph{transaction fees} and an \emph{inflation mechanism}.

We may enforce a transaction owner to pay a fee for any transaction it submits. 
In the current protocol, this can be implemented by adding an auxiliary transaction recipient $\perp$ in the transfer function $\tau$. 
Formally, $tx.\tau(\perp) = \textit{fee}(tx)$, where $\textit{fee} : \tT \rightarrow \Nat_0$ is some predefined function.
It might be reasonable to incentivize larger transactions over small ones as it will lower the system load.


%
Similarly, we can implement an inflation mechanism. 
Each participant can be allowed to periodically receive newly created funds.
Periods of time can be represented with the help of verifiable histories. 
If a local verifiable history $h$ of process $p$ is such that $\exists C_i, C_j \in h: C_i \sqsubset C_j$ and  $|height(C_j) -height(C_i)| > \nu$, where $\nu$ is some predefined sufficiently large constant and $p \in \textit{members}(C_i)$, then $p$ can issue a special transaction $tx_{inf}$ that contains information about the history $h$ and references configurations $C_i$ and $C_j$. 
A transaction of this type passes through the \Name pipeline with an adjusted Transaction Validation.  
For example, in order to prevent processes from using the same history $h$ and configurations $C_i$ and $C_j$, we should mark two such inflation transactions $tx_{1}$ and $tx_{2}$ by the same owner as conflicting if the intervals defined by the corresponding $C_i$ and $C_j$ intersect: $[tx_{1}.C_i; tx_{1}.C_j] \cap [tx_{2}.C_i; tx_{2}.C_j] \neq \emptyset$.
The amount of funds $\gamma$ a process obtains from such a transaction might be proportional to its minimum stake on the set of transactions $C$ such that $C \in h$ and $C_i \sqsubseteq C \sqsubseteq C_j$.

We may require that a process $p$ must delegate its stake to sufficiently ``rich''  validators to collect the reward or create money: to engage in these mechanisms, suitable validators should hold at least fraction $\chi$ of total stake. 
This is a possible economical incentive for keeping the number of validators relatively small and thus increase the system efficiency.

The analysis and rationale of selecting function $\textit{fee}$ and suitable parameters $\nu, \gamma, \chi$ 
as well as implementing a system that rewards actively participating validators
is an interesting avenue for future work.

\myparagraph{Forward-secure digital signatures.} 
In Section~\ref{sec:model}, we already mentioned that known implementations of forward-secure signatures can be divided into two groups: (1) the ones where parameter $T$ is bounded and fixed in advance and (2) the ones that do not oblige to fix $T$ in advance, but the computational time is unbounded and depends on the timestamp.

The first possible approach is to use forward-secure signatures implementations that require $T$ to be fixed in advance and choose a sufficiently large number, e.g. $2^{64}$, which bounds the total number of possible transactions. 
Recall that the number of transactions confirmed on Bitcoin~\citep{bitcoin} blockchain is less than 500~million for 12 years, which suggests that $T=2^{64}$ should be more than enough.

The second approach is based on the use of forward-secure signatures schemes with unbounded $T$.
To solve an issue with indefinite local computational time, we might change public keys at some moment of time (e.g., when the height of some installed configuration is sufficiently large). 
Process $q$ might issue a special transaction that indicates changing $q$'s forward-secure public key. 
Such transactions should pass through the whole \Name protocol with the Transaction Validation block adjusted for them.
The key replacement that is done with the help of transactions ensures that a process does not send different public keys to different users.
The initial timestamp can be shifted to the height of some recently installed configuration without loss of efficiency.
%
%
Notice, however, that local computational time grows only slowly (logarithmically) with $T$~\citep{MaMiMi02}, so this might not matter in practice.

\end{document}